\documentclass[journal,final,twocolumn,10pt,twoside]{IEEEtranTCOM}

\normalsize

\usepackage{mathtools}
\usepackage{amsmath}
\usepackage{amssymb}
\usepackage{amsmath}
\usepackage{cite}
\usepackage{algorithm, algorithmic}

\usepackage{siunitx}
\usepackage{multirow}
\usepackage{multicol}
\usepackage{color}
\usepackage{xcolor}

\usepackage{graphicx}
\usepackage{graphicx}
\date{}

\usepackage{calc}
\newcolumntype{M}[1]{>{\centering\arraybackslash}m{#1}}
\newcolumntype{N}{@{}m{0pt}@{}}

\usepackage{mdwmath}
\usepackage{blindtext}
\usepackage{eqparbox}
\usepackage{fixltx2e}
\usepackage{stfloats}
\DeclareMathOperator{\E}{\mathbb{E}}

\usepackage{amsthm}

\newtheorem{theorem}{{Theorem}}
\newtheorem{lemma}[theorem]{{Lemma}}
\newtheorem{proposition}[theorem]{{Proposition}}

\DeclareMathAlphabet{\mathbfsl}{OT1}{ppl}{b}{it} 

\newcommand{\be}[1]{\begin{equation}\label{#1}}
	\newcommand{\ee}{\end{equation}}

\renewcommand{\leq}{\leqslant}
 
\renewcommand{\geq}{\geqslant}

\renewcommand{\Bbb}{\mathbb}

\newcommand{\R}{{\Bbb R}}

\newcommand{\Tref}[1]{Theo\-rem\,\ref{#1}}
\newcommand{\Pref}[1]{Pro\-po\-si\-tion\,\ref{#1}}
\newcommand{\Lref}[1]{Lem\-ma\,\ref{#1}}
\newcommand{\Cref}[1]{Co\-ro\-lla\-ry\,\ref{#1}}

\begin{document}
	\title{Covert Millimeter-Wave Communication:\\
	Design Strategies and Performance Analysis}
	\author{Mohammad Vahid Jamali and Hessam Mahdavifar,~\IEEEmembership{Member,~IEEE}
		\thanks{This paper was presented in part at the IEEE Global Communications Conference (GLOBECOM), Waikoloa, HI, USA, Dec. 2019 \cite{jamali2019covert}.}
		\thanks{The authors are with the Department of Electrical Engineering and Computer Science, University of Michigan, Ann Arbor, MI 48109, USA (e-mail: mvjamali@umich.edu, hessam@umich.edu).}
		\thanks{This work was supported by the National Science Foundation under grants CCF--1763348, CCF--1909771, and CCF--1941633.}
	}
\maketitle
\begin{abstract}
In this paper, we investigate covert communication over millimeter-wave (mmWave) frequencies. In particular, a mmWave transmitter, referred to as Alice, attempts to reliably communicate to a receiver, referred to as Bob, while hiding the existence of communication from a warden, referred to as Willie.
 In this regard, operating over the mmWave bands not only increases the covertness thanks to directional beams, but also increases the transmission data rates given much more available bandwidths and enables ultra-low form factor transceivers due to the lower wavelengths used compared to the conventional radio frequency (RF) counterpart. We first assume that the transmitter Alice employs two independent antenna arrays in which one of the arrays is to form a directive beam for data transmission to Bob. The other antenna array is used by Alice to generate another beam toward Willie as a jamming signal while changing the transmit power independently across the transmission blocks in order to achieve the desired covertness. For this dual-beam setup, we characterize Willie's detection error rate with the optimal detector and the closed-form of its expected value from Alice's perspective. We then derive the closed-form expression for the outage probability of the Alice-Bob link, which enables characterizing the optimal covert rate that can be achieved using the proposed setup. We further obtain tractable forms for the ergodic capacity of the Alice-Bob link involving only one-dimensional integrals that can be computed in closed forms for most ranges of the channel parameters. Finally, we highlight how the results can be extended to more practical scenarios, particularly to the cases where perfect information about the location of the passive warden is not available. Our results demonstrate the advantages of covert mmWave communication compared to the RF counterpart. The research in this paper is the first analytical attempt in exploring covert communication using mmWave systems.
\end{abstract}
\begin{keywords} 
Covert communication, mmWave communication, communication with low probability of detection, detection error rate, effective covert rate, ergodic capacity, Nakagami fading channels.
\end{keywords}
\section{Introduction}\label{Sec1}
\IEEEPARstart{R}apid growth of wireless networks and the emergence of variety of applications, including the  Internet of Things (IoT), massive machine-type communication (mMTC), and critical controls, necessitate sophisticated solutions to secure data transmission. Traditionally, the main objective of security schemes, using either cryptographic or information-theoretic approaches, is to secure data in the presence of adversary eavesdroppers. However, a stronger level of security can be obtained in wireless networks if the existence of communication is hidden from the adversaries. To this end, recently, there has been increasing attention to investigate covert communication, also referred to as communication with \textit{low probability of detection} (LPD), in various scenarios with the goal of hiding the existence of communication \cite{bash2013limits, bash2015hiding, che2013reliable, wang2016fundamental, bloch2016covert,tahmasbi2018first, arumugam2019embedding,arumugam2018covert}.  Generally speaking, covert communication refers to the problem of reliable communication between a transmitter Alice and a receiver Bob while maintaining a low probability of detecting the existence of communication from the perspective of a warden Willie \cite{bash2015hiding}.

{In contrast to traditional cryptographic schemes and similar to physical-layer security schemes, covert communication exploits the physical layer of a communication network to provide security. The most important difference in the setting of physical-layer security and covert communication is the functionality of the illegitimate parties, i.e., the eavesdropper Eve and the warden Willie. In fact, while covert communication attempts to hide the existence of the communication from the warden, physical-layer security schemes aim at minimizing the information obtained by the eavesdroppers through exploiting the dynamic characteristics of the wireless medium \cite{bloch2011physical_book}. Therefore, as opposed to covert communication, physical-layer security does not provide protection against the detection of a transmission. Hence, covert communication can provide a stronger level of security while also achieving privacy of the transmitter by guaranteeing a negligible detection probability of the transmission at a warden.}

{The information-theoretic limits on the rate of covert communication have been presented in \cite{bash2013limits} for additive white Gaussian noise (AWGN) channels. More specifically, assuming the communication blocklength to be $n$, it has been proved in \cite{bash2013limits} that $\mathcal{O}(\sqrt{n})$ bits of information can be transmitted to Bob, reliably and covertly, in $n$ uses of the channel, as $n\to\infty$.}
The same square root law has been developed for binary symmetric channels (BSCs) in \cite{che2013reliable} and discrete memoryless channels (DMCs) in \cite{wang2016fundamental}. Moreover, the principle of channel resolvability has been  used in \cite{bloch2016covert} to develop a coding scheme that can reduce the number of required shared key bits. Also, the first- and second-order asymptotics
of covert communication over binary-input DMCs have been studied in \cite{tahmasbi2018first}. The covert communication setup has also been extended to broadcast channels \cite{arumugam2019embedding} and to multiple-access channels \cite{arumugam2018covert} from an information-theoretic perspective.

{The achievable covert rate (i.e., the ratio of the number of information bits to the number of channel uses) in the aforementioned framework is zero as $n$ grows large since $\lim_{n\to\infty}\mathcal{O}(\sqrt{n})/n=0$.} However, it is demonstrated that positive covert rates can be achieved by introducing additional uncertainty, from Willie's perspective, into the system. For instance, it is shown in \cite{lee2015achieving,goeckel2016covert} that Willie's uncertainty about his noise power helps achieving positive 
covert rates. Moreover, by considering slotted AWGN channels, it is proved in \cite{bash2016covert} that positive covert rates are achievable if the warden does not know when the transmission is taking place. The possibility of achieving positive-rate covert communication is further demonstrated for several other scenarios such as amplify-and-forward (AF) relaying networks with a greedy relay attempting to transmit its own information to the destination on top of forwarding the source's information \cite{hu2018covert}, dual-hop relaying systems with channel uncertainty \cite{wang2019covert}, a downlink scenario under channel uncertainty and with a legitimate user as the cover \cite{shahzad2017covert}, and a single-hop setup with a full-duplex receiver acting as a jammer \cite{shahzad2018achieving}. 
Additionally, covert communication in the presence of a multi-antenna adversary, under delay constraints, and for the case of quasi-static wireless fading channels is considered in \cite{shahzad2019covert}.
 In \cite{hu2019covert}, channel inversion power control is adopted to achieve covert communication with the aid of a full-duplex receiver. Covert communication in the context of unmanned aerial vehicle (UAV) networks is considered in \cite{zhou2019joint}. { Physical-layer security has been investigated in \cite{wang2018physical,liu2020robust} for visible light communication (VLC) and can be extended to covert communication.}
 	 Very recently, the problem of joint covert communication and secure transmission in untrusted relaying networks in the presence of multiple wardens has been considered in \cite{forouzesh2020covert}. Moreover, the benefits of beamforming in improving the performance of covert communication in the presence of a jammer has been studied in \cite{forouzesh2020covert2}.

Prior studies on covert communication in wireless networks mostly consider omni-directional transmission over conventional radio frequency (RF) wireless links. However, a superior performance can be potentially attained when performing the covert communication over the millimeter-wave (mmWave) bands. In particular, operating over the mmWave bands not only increases the covertness thanks to directional beams, but also increases the transmission data rates given much more available bandwidths and enables ultra-low form factor transceivers due to the lower wavelengths used compared to the conventional RF counterpart.
This makes the mmWave band a suitable option for covert communication to increase the security level of wireless applications involving critical data. Also, with the advancement in mmWave communications and rapid development of mmWave cellular networks in the fifth generation of wireless networks (5G) and beyond that, mmWave systems will serve as major components for a wide range of emerging wireless networking applications and use cases. This necessitates secure transmission schemes for mmWave systems and further highlights the importance of covert mmWave communication.

The channel model and system architecture of mmWave communication systems significantly differ from those of RF communication. In particular, communication over the mmWave bands can exploit directive beams, thanks to the deployment of massive antenna arrays, to compensate for the significant path loss over this range of frequency\footnote{{Directive beams can also be exploited over RF systems through beamforming technology. However, given much smaller wavelengths at the mmWave bands compared to the RF bands, it is much easier to realize large antenna arrays and (narrow) directive beams, especially at mobile users, over mmWave systems.}}. Meanwhile, the significant susceptibility of directive mmWave links to blockage results in a nearly bimodal channel depending on whether a line-of-sight (LOS) link exists between the transmitter and receiver \cite{andrews2017modeling}. Furthermore, the properties of mmWave and RF channels, including path loss and statistical distribution of fading, are often modeled very differently. Therefore, the existing results on covert communication cannot immediately be extended to covert communication over the mmWave bands. 

In this paper, we study covert communication over mmWave channels from a communication theory perspective. More specifically, we analyze the performance of the system in the limit as the blocklength $n$ grows large. In order to achieve a positive-rate covert communication, the transmitter Alice is equipped with two antenna arrays each pointed to a different direction and carrying independent data streams. The first antenna array forms a directive beam for covert data transmission to Bob. The second array is used to generate another beam toward Willie as a jamming signal while the transmit power is changed independently across the transmission blocks in order to achieve desired covertness. The research in this paper is the first attempt in analytical studies of covert communication over mmWave systems. It is worth mentioning that a conceptual framework for covert mmWave communication was envisioned in \cite{cotton2009millimeter} without providing analytical studies. To the best of the authors' knowledge, no analytical characterization for covert mmWave system has been carried out in prior works. 
Very recently, after the appearance of the initial version of this work \cite{jamali2019covert}, Zhang \textit{et al.} \cite{zhang2020joint} studied joint beam training and data transmission for covert mmWave communication. More specifically, the authors of \cite{zhang2020joint} aimed at jointly optimizing the beam training duration (to establish a directional link between Alice and Bob), the training power, and the data transmission power to maximize the effective covert rate while satisfying the covertness constraint on Willie. 

{
The main contributions of the paper are summarized as follows.
\begin{itemize}
	\item We characterize Willie's optimal detection performance in terms of the overall (minimum) detection error rate, and derive the closed form for the expected value of the detection error rate from Alice's perspective. 
	\item To characterize the performance of the desired link, we obtain the closed-form expression for the outage probability of the Alice-Bob link, and then formulate the optimal covert rate that is achievable in our proposed setup.
	\item We further obtain tractable forms for the ergodic capacity of the Alice-Bob link involving only one-dimensional integrals that can be computed in closed forms for most ranges of the channel parameters.
	\item We highlight how the results of the paper can be extended to more practical scenarios, particularly to the cases where perfect information about Willie's location is not available to Alice. 
	We also provide several important directions for future research on covert mmWave communication.
	\item We present extensive numerical analysis to study the system performance in various aspects.
\end{itemize}
}

The rest of the paper is organized as follows. In Section \ref{Sec2}, we briefly summarize the mmWave channel model and describe the proposed covert mmWave communication setup. In Section \ref{Sec3}, we analyze Willie's overall error rate with an optimal radiometer detector, and then obtain its expected value from Alice's perspective. Section \ref{Sec4} is devoted to studying the performance of the Alice-Bob link, in terms of the outage probability, effective covert rate, and ergodic capacity. Discussions about various realistic scenarios, including imperfect knowledge about Willie's location, as well as some future research directions are provided in Section \ref{Sec5}. Finally, extensive numerical results are presented in Section \ref{Sec6}, and the paper is concluded in Section \ref{Sec7}.

\section{Channel and System Models}\label{Sec2}
{ In this section, we first briefly characterize mmWave channels and describe their distinct properties to enable explaining the system model presented afterwards.
\subsection{MmWave Channel Model}\label{Sec2A}}
Recent experimental studies have demonstrated that mmWave links are highly sensitive to blocking effects \cite{rappaport2013millimeter,andrews2017modeling}. In order to model this characteristic, a proper channel model should differentiate between the LOS and non-LOS (NLOS) channel models. Therefore, two different sets of parameters are considered for the LOS and NLOS mmWave links, and a deterministic function $P_{\rm LOS}(d_{ij}) \in [0,1]$, that is a non-increasing function of the link length $d_{ij}$ (in meters) between the nodes $i$ and $j$, is defined to characterize the probability of an arbitrary link of length $d_{ij}$ being LOS. 
In this paper,	we consider a generic function $P_{\rm LOS}(d_{ij})$ throughout our analysis and use the model $P_{\rm LOS}(d_{ij})={\rm e}^{-d_{ij}/200}$, suggested in \cite{andrews2017modeling}, for our numerical analysis.

Next, we briefly describe how the LOS and NLOS channels can be characterized.
Similar to \cite{andrews2017modeling}, we express the channel coefficient for an arbitrary mmWave link between the transmitter $i$ and receiver $j$ as $h_{ij}=\tilde{h}_{ij}\sqrt{{G}_{ij}L_{ij}}$, where $\tilde{h}_{ij}$, ${G}_{ij}$, and $L_{ij}$ are the channel fading coefficient, the total directivity gain (including both the transmitter and the receiver beamforming gains), and the path loss of the $i$-$j$ mmWave link, respectively. This model is widely used in the literature for analytical tractability purposes. The reader is referred to \cite{andrews2017modeling} and the references therein for more details on mmWave channel modeling and also the validity of this model.

\begin{table}[t]
	\centering
	\caption{Probability mass function (PMF) of the directivity gain of a node $q$ with beamsteering error \cite{di2015stochastic}.}
	\label{T1}
	\begin{tabular}{M{0.6cm}||M{2.1cm}M{3cm}}  
		$k$ & $1$ & $2$\\ \hline\hline  {\vspace{0.1cm}}
		$g^{(q)}_{k}$ {\vspace{0.05cm}} & $M^{(q)}_{\mathcal{X}}$ & $m^{(q)}_{\mathcal{X}}$ \\ \hline {\vspace{0.1cm}}
$b^{(q)}_{k}$ {\vspace{0.2cm}} &$F_{|\mathcal{E}^{(q)}_{\mathcal{X}}|}\left({\theta^{(q)}_{\mathcal{X}}}/{2}\right)$& $1-F_{|\mathcal{E}^{(q)}_{\mathcal{X}}|}\left({\theta^{(q)}_{\mathcal{X}}}/{2}\right)$	\\ \hline 
	\end{tabular}
\end{table}

{To characterize the path loss $L_{ij}$ of the $i$-$j$ link with the length $d_{ij}$, we consider different path loss exponents ($\alpha_{\rm L},\alpha_{\rm N}$) and intercepts ($C_{\rm L},C_{\rm N}$) for the LOS and NLOS links, respectively. Let $L^{(\rm L)}_{ij}$ and $L^{(\rm N)}_{ij}$ denote the path losses of the LOS and NLOS links, respectively. Then the path loss $L_{ij}$ is either equal to $L^{(\rm L)}_{ij}=C_{\rm L}d_{ij}^{-\alpha_{\rm L}}$ with probability $P_{\rm LOS}(d_{ij})$ or equal to $L^{(\rm N)}_{ij}=C_{\rm N}d_{ij}^{-\alpha_{\rm N}}$ with probability $1-P_{\rm LOS}(d_{ij})$.
Note that the path loss in the NLOS links can be much higher than that of the LOS path due to the weak diffractions in the mmWave bands \cite{rappaport2013millimeter}.

To ascertain the total directivity gain ${G}_{ij}$, we use the common sectored-pattern antenna model \cite{bai2015coverage,di2015stochastic} which approximates the actual array beam pattern by a step function, i.e., with a constant main lobe gain $M^{(q)}_{\mathcal{X}}$ over the beamwidth $\theta^{(q)}_{\mathcal{X}}$ and a constant side lobe gain $m^{(q)}_{\mathcal{X}}$ otherwise, where ${\mathcal{X}}\in\{{\rm TX,RX}\}$ and ${q}\in\{{i,j}\}$. Then, for a given link, if the spatial arrangement of the beams of the transmitter and receiver are known, the total directivity gain can be obtained from the product of the gains of the transmitter and receiver. If the main lobe of a node $q$ (either transmitter or receiver) is pointed to another node, 
we assume that an additive beamsteering error exists, denoted by a symmetric random variable (RV) $\mathcal{E}^{(q)}_{\mathcal{X}}$, in the vicinity of the transmitter-receiver direction. Same as in \cite{di2015stochastic}, it is assumed that node $q$ has a gain equal to $M^{(q)}_{\mathcal{X}}$ if $|\mathcal{E}^{(q)}_{\mathcal{X}}|<\theta^{(q)}_{\mathcal{X}}/2$, which occurs with probability $F_{|\mathcal{E}^{(q)}_{\mathcal{X}}|}(\theta^{(q)}_{\mathcal{X}}/2)$ with $F_X(x)$ being the cumulative distribution function (CDF) of the RV $X$. Otherwise, it has a gain equal to $m^{(q)}_{\mathcal{X}}$. 
Then the probability mass function (PMF) of the directivity gain of a node $q$ with beamsteering error can be expressed as a RV taking the values $g_k^{(q)}$ with probabilities $b_k^{(q)}$, $k\in\{1,2\}$, as summarized in Table \ref{T1}.

Finally, it is common in the literature to model the fading amplitude of mmWave links as independent Nakagami-distributed RVs with shape parameter $\nu\geq 1/2$ and scale parameter $\Omega=\E[|{\tilde{h}}_{ij}|^2]=1$, and consider different Nakagami parameters for the LOS and NLOS links as $\nu_{\rm L}$ and $\nu_{\rm N}$, respectively \cite{andrews2017modeling,bai2015coverage}. In the case of Nakagami-$m$ fading with parameters $\nu_{\mathcal{B}}$, $\mathcal{B}\in\{{\rm L},{\rm N}\}$, and $\Omega=1$, $|{\tilde{h}}_{ij}|^2$ has a normalized gamma distribution with shape and scale parameters of $\nu_{\mathcal{B}}$ and $1/\nu_{\mathcal{B}}$, respectively.
Therefore, the probability density function (PDF) of $|{\tilde{h}}_{ij}|^2$ is given by \cite{simon2005digital}
\begin{align}\label{pdf_gamma}
f_{|{\tilde{h}}_{ij}|^2}(y)=\frac{ {\nu_{\mathcal{B}}}^{\nu_{\mathcal{B}}} y^{\nu_{\mathcal{B}}-1}}{\Gamma(\nu_{\mathcal{B}})}\exp\left(-\nu_{\mathcal{B}}y\right).
\end{align}
As it will be clarified later, in order to derive tractable closed-form expressions, we will often assume in this paper that the shape parameter $\nu_{\mathcal{B}}$ is an integer.

Note that, from an information-theoretic perspective, mmWave communications, and in general wideband communications under power constraints, can be viewed as low-capacity scenarios \cite{fereydounian2018channel,jamali2020massive,jamali2018low} suggesting a natural framework for covert mmWave communication.}

\subsection{System Model}\label{Sec2B}
We consider the well-known setup for covert communication comprised of three parties: a transmitter Alice is intending to covertly communicate to a receiver Bob over the mmWave bands when a warden Willie is attempting to  detect the existence of this communication. Alice employs a dual-beam mmWave transmitter consisting of two antenna arrays. The first antenna array is used for the transmission to Bob while the second array is exploited as a jammer to enable positive-rate covert mmWave communication. 
{ Note that, although the number of antenna elements in mmWave arrays is typically large to compensate for the significant propagation loss through beamforming (directionality gain), the wavelengths are much smaller than that of the RF communication (e.g., $5$ \si{mm} at $60$ \si{GHz} versus $60$ \si{mm} at $5$ \si{GHz}). Therefore, it is feasible to realize large mmWave antenna arrays in a small package thanks to the recent advancements in antenna circuit design \cite{andrews2017modeling,hong2014study}. Additionally, it is practical to consider two separate antenna arrays in a given mmWave transmitter. In particular, a first-of-the-kind mmWave antenna system prototype has been presented in \cite{hong2014study} that integrated two separate mmWave antenna arrays, each of size $1\times 16$, inside a Samsung cell phone (one at the top and the other at the bottom of the cell phone). Moreover, the authors in \cite{rangan2014millimeter} proposed incorporating several mmWave antenna arrays throughout a mobile device to provide path diversity from blockage by human obstructions.}

Given the above transmission model, when Bob is not in the main lobe of the Alice-Willie link, he receives the jamming signal gained with the side lobe of the second array in addition to receiving the desired signal from Alice with the main lobe of the first array. Similarly, when Willie is not in the main lobe of the Alice-Bob link, he receives the desired signal gained with the side lobe of the first array in addition to receiving the jamming signal from Alice with the main lobe of the second array. On the other hand, when Bob is in the main lobe of the Alice-Willie link (or equivalently, Willie is in the main lobe of the Alice-Bob link), both of the received signals by Bob and Willie are gained with main lobes\footnote{{In such extreme cases, both Bob and Willie receive the desired signal with the same gain from Alice. Given the small wavelengths at the mmWave bands, one can exploit relatively large antenna arrays to realize three-dimensional (3D) beamforming \cite{razavizadeh2014three,forouzesh2020covert2} at Alice's first array to further focus the beam (in three dimensions) toward Bob and reduce the chance of Willie receiving the desired signal with the (large) main lobe gain. Further investigation on this direction is left for future work.}}. Throughout our analysis in Sections \ref{Sec3} and \ref{Sec4}, we assume that Alice, Bob, and Willie are in some fixed locations (hence, having some given directivity gains). And we leave the discussion about various realistic scenarios, such as imperfect knowledge of Willie's location, to Section \ref{Sec5}. 


Let the channel coefficients between Alice's first and second arrays and the node $j\in\{b,w\}$ (representing Bob and Willie) be denoted by $h_{aj,f}$ and $h_{aj,s}$, respectively. Then it can be observed that the path loss gains are the same, i.e., $L_{aj,f}=L_{aj,s}\triangleq L_{aj}$, while the fading gains $|\tilde{h}_{aj,f}|^2$ and $|\tilde{h}_{aj,s}|^2$ are independent normalized gamma RVs\footnote{{Note that the fading coefficients can be considered uncorrelated if the antenna arrays are spaced more than half a wavelength \cite{molisch2012wireless}. Given that the wavelengths are very small at the mmWave bands, e.g., $5$ \si{mm} at $60$ \si{GHz}, it is easy to realize tens of wavelengths of spacing between the arrays and ensure independence between the fading coefficients.}}. We assume quasi-static fading channels meaning that fading coefficients remain constant over a block of $n$ channel uses. We further assume that Alice transmits the desired signal with a publicly-known power $P_a$ while the jamming transmit power $P_J$ of the second array is not known and is changed independently across transmission blocks. In this paper, we assume that $P_J$ is drawn from a uniform distribution over the interval $[0,P_J^{\rm max}]$ while the results can be extended to other distributions using a similar approach. Let $G_{aj,f}$ and $G_{aj,s}$ denote the total directivity gains of the links between Alice's first and second arrays and the node $j\in\{b,w\}$, respectively. Then,
 the received signals by Bob and Willie at each channel use $i$, for $i=1,2,...,n$, are given by
\begin{align}
\mathbf{y}_{b}(i)=&\sqrt{P_aG_{ab,f}L_{ab}}~\tilde{h}_{ab,f}\mathbf{x}_a(i)\label{yb}\nonumber\\
&+\sqrt{P_JG_{ab,s}L_{ab}}~\tilde{h}_{ab,s}\mathbf{x}_J(i)+\mathbf{n}_b(i),\\
\mathbf{y}_{w}(i)=&\sqrt{P_aG_{aw,f}L_{aw}}~\tilde{h}_{aw,f}\mathbf{x}_a(i)\label{yw}\nonumber\\
&+\sqrt{P_JG_{aw,s}L_{aw}}~\tilde{h}_{aw,s}\mathbf{x}_J(i)+\mathbf{n}_w(i),
\end{align}
respectively, where $\mathbf{x}_a$ and $\mathbf{x}_J$ are the desired signal and the jamming signal, respectively, each having a zero-mean Gaussian distribution satisfying $\E[|\mathbf{x}_a(i)|^2]=\E[|\mathbf{x}_J(i)|^2]=1$. Moreover, $\mathbf{n}_b$ and $\mathbf{n}_w$ are zero-mean Gaussian noise components at Bob and Willie's receivers with  variances $\sigma^2_b$ and $\sigma^2_w$, respectively. 

Finally, note that the results derived in this paper can be applied to a similar system model, though with Rayleigh fading channels, by substituting $\nu_{\mathcal{B}}=1$. This is because the normalized gamma distribution simplifies to the exponential distribution with mean one in the special case of $\nu_{\mathcal{B}}=1$.

\section{Willie's Detection Error Rate}\label{Sec3}
As discussed earlier, Willie's goal is to detect whether Alice is transmitting to Bob or not. It is assumed that Willie has a perfect knowledge about the channel between himself and Alice, and applies binary hypothesis testing while being unaware of the value of $P_J$.
{The null hypothesis $\mathbb{H}_0$ states that Alice did not transmit to Bob, and the alternative hypothesis $\mathbb{H}_1$ specifies that a transmission from Alice to Bob occurred. Willie's decision of hypothesis $\mathbb{H}_1$ when $\mathbb{H}_0$ is true is referred to as a \textit{false alarm} and its probability is denoted by $P_{\rm FA}$. Moreover, Willie's decision in favor of $\mathbb{H}_0$ when $\mathbb{H}_1$ is true is referred to as a \textit{missed detection} with the probability denoted by $P_{\rm MD}$. Then Willie's overall detection error rate is defined as $P_{e,w}\triangleq P_{\rm FA}+P_{\rm MD}$. }
We say that a positive-rate covert communication is possible if for any $\epsilon>0$ there exists a positive-rate communication between Alice and Bob satisfying $P_{e,w}\geq1-\epsilon$ as the number of channel uses $n\to\infty$. 
{ In this section, we first derive the minimum value of $P_{e,w}$, denoted by $P_{e,w}^*$, under the assumption of complete knowledge of the channels and an optimal radiometer detector at Willie. We also assume that Willie observes infinitely large number of channel uses. It is worth mentioning that such assumptions correspond to the worst-case scenario for the covertness requirement as they result in the minimum error rate for Willie. We then derive the closed-form expression of the expected value of $P_{e,w}^*$ form Alice's perspective in Section \ref{Sec3B}.}
\subsection{$P_{e,w}$ with the Optimal Detector at Willie}\label{Sec3A}
As it is proved in \cite[Lemma 2]{sobers2017covert} for AWGN channels and also pointed out in \cite[Lemma 1]{shahzad2017covert}, the optimal decision rule that minimizes Willie's detection error
  is given by
\begin{align}\label{radiometer}
T_w\triangleq\frac{1}{n}\sum_{i=1}^{n}|\mathbf{y}_{w}(i)|^2
\underset{\mathbb{H}_0}{\overset{\mathbb{H}_1}{\gtrless}}\tau,
\end{align}
where $\tau$ is Willie's detection threshold for which we obtain the corresponding optimal value/range later in this subsection. Using \eqref{yw} and the definition of $T_w$ in \eqref{radiometer}, we can write $T_w$ under hypothesis $\mathbb{H}_0$, denoted by $T_w^{\mathbb{H}_0}$, as
\begin{align}\label{TH0_1}
T_w^{\mathbb{H}_0}=\left(P_JG_{aw,s}L_{aw}|\tilde{h}_{aw,s}|^2+\sigma^2_w\right)\frac{\chi^2_{2n}}{n},
\end{align}
where $\chi^2_{2n}$ denotes a chi-squared RV with $2n$ degrees of freedom. According to the strong law of large numbers, $\frac{\chi^2_{2n}}{n}$ converges to $1$, \textit{almost surely}, as $n\to\infty$. Therefore, using Lebesgue's dominated convergence theorem \cite{browder2012mathematical}, we cam replace $\frac{\chi^2_{2n}}{n}$ by $1$ to rewrite $T_w^{\mathbb{H}_0}$ as
\begin{align}\label{TH0}
T_w^{\mathbb{H}_0}=P_JG_{aw,s}L_{aw}|\tilde{h}_{aw,s}|^2+\sigma^2_w.
\end{align}
Similarly, $T_w$ under hypothesis $\mathbb{H}_1$ can be obtained as
\begin{align}\label{TH1}
T_w^{\mathbb{H}_1}&\!=\!P_aG_{aw,f}L_{aw}|\tilde{h}_{aw,f}|^2\!+\!P_JG_{aw,s}L_{aw}|\tilde{h}_{aw,s}|^2\!+\!\sigma^2_w.
\end{align}

{ One can observe that if Willie has a complete knowledge about the jamming power $P_J$, he can choose any threshold in the interval $T_w^{\mathbb{H}_0}\leq \tau \leq T_w^{\mathbb{H}_1}$ to achieve $P_{\rm FA}=P_{\rm MD}=0$, and hence $P_{e,w}=0$ (recall that we assumed Willie has full knowledge about the realization of the Alice-Willie channel to constitute a worst-case scenario for the covertness requirement). Alternatively, if $P_J$ is known to Willie with a probability $q>0$, then we cannot satisfy the covertness requirement for the $\epsilon$ values smaller than $q$. In other words, some sort of randomness is required in the system model to enable covert communication. In the following theorem, we characterize the optimal threshold of Willie's detector and its corresponding minimum detection error rate under the assumption that $P_J$ is completely unknown to Willie and changes randomly per transmission block according to a uniform distribution over the interval $[0,P_J^{\rm max}]$.}
\begin{theorem}\label{Thm1}
The optimal threshold $\tau^*$ for Willie's detector is in the interval
\begin{align}\label{tau*}
{\tau^*\in[\min(\lambda_1, \lambda_2), \max(\lambda_1, \lambda_2)],}
\end{align}
and  the corresponding minimum detection error rate is 
\begin{align}\label{pe}
P_{e,w}^*=\left\{\begin{matrix}
\hspace{-2.9cm}0, &\lambda_1<\lambda_2, \\ 
1-\frac{P_aG_{aw,f}|\tilde{h}_{aw,f}|^2}{P_J^{\rm max}G_{aw,s}|\tilde{h}_{aw,s}|^2}, &\lambda_1\geq\lambda_2,
\end{matrix}\right.
\end{align}
where $\lambda_1\triangleq P_J^{\rm max}G_{aw,s}L_{aw}|\tilde{h}_{aw,s}|^2+\sigma^2_w$ and  $\lambda_2\triangleq P_aG_{aw,f}L_{aw}|\tilde{h}_{aw,f}|^2+\sigma^2_w$.
\end{theorem}
\begin{proof}
Using \eqref{TH0}, the false alarm probability is given by
\begin{align}\label{pfa}
P_{\rm FA}&=\Pr\left(T_w^{\mathbb{H}_0}>\tau\right)=\Pr\left(P_J>\frac{\tau-\sigma^2_w}{G_{aw,s}L_{aw}|\tilde{h}_{aw,s}|^2}\right)\nonumber\\
&=\left\{\begin{matrix}
\hspace{-3.5cm}1, &\tau<\sigma^2_w, \\ 
1-\frac{\tau-\sigma^2_w}{P_J^{\rm max}G_{aw,s}L_{aw}|\tilde{h}_{aw,s}|^2}, &\sigma^2_w\leq\tau\leq\lambda
_1,\\
\hspace{-3.5cm}0,&\tau\geq\lambda_1.
\end{matrix}\right.
\end{align}
Also, by \eqref{TH1} the missed detection probability is given by
\begin{align}\label{pmd}
P_{\rm MD}&=\Pr\left(T_w^{\mathbb{H}_1}<\tau\right)=\Pr\left(P_J<\frac{\tau-\lambda_2}{G_{aw,s}L_{aw}|\tilde{h}_{aw,s}|^2}\right)\nonumber\\
&=\left\{\begin{matrix}
\hspace{-2.8cm}0, &\tau<\lambda_2, \\ 
\frac{\tau-\lambda_2}{P_J^{\rm max}G_{aw,s}L_{aw}|\tilde{h}_{aw,s}|^2}, &\lambda_2\leq\tau\leq\lambda
_3,\\
\hspace{-2.8cm}1,&\tau\geq\lambda_3,
\end{matrix}\right.
\end{align}
where $\lambda_3\triangleq \lambda_2+P_J^{\rm max}G_{aw,s}L_{aw}|\tilde{h}_{aw,s}|^2$. 
Next, we consider the following two cases.

\textit{Case I:} When $\lambda_1<\lambda_2$, Willie's receiver can choose any thresholds in the interval $[\lambda_1,\lambda_2]$ to get both $P_{\rm FA}=0$ and $P_{\rm MD}=0$, resulting in zero detection error $P_{e,w}\triangleq P_{\rm FA}+P_{\rm MD}$.

\textit{Case II:} When $\lambda_1\geq\lambda_2$, we can write the overall detection error rate $P_{e,w}\triangleq P_{\rm FA}+P_{\rm MD}$, using \eqref{pfa} and \eqref{pmd}, as
\begin{align}\label{Pew}
P_{e,w}=\left\{\begin{matrix}
\hspace{-3.45cm}1, &\tau\leq\sigma^2_w, \\ 
1-\frac{\tau-\sigma^2_w}{P_J^{\rm max}G_{aw,s}L_{aw}|\tilde{h}_{aw,s}|^2}, &\sigma^2_w\leq\tau\leq\lambda
_2,\\
\hspace{-0.52cm}1-\frac{P_aG_{aw,f}|\tilde{h}_{aw,f}|^2}{P_J^{\rm max}G_{aw,s}|\tilde{h}_{aw,s}|^2}, &\lambda_2\leq\tau\leq\lambda
_1,\\
\hspace{-0.6cm}\frac{\tau-\lambda_2}{P_J^{\rm max}G_{aw,s}L_{aw}|\tilde{h}_{aw,s}|^2}, &\lambda_1\leq\tau\leq\lambda
_3,\\
\hspace{-3.45cm}1,&\tau\geq\lambda_3.
\end{matrix}\right.
\end{align}
Therefore, based on \eqref{Pew}, the receiver never chooses $\tau\leq\sigma^2_w$ or $\tau\geq\lambda_3$ since they result in the worst performance $P_{e,w}=1$. Moreover, \eqref{Pew} monotonically decreases, with respect to $\tau$, in the interval $\sigma^2_w\leq\tau\leq\lambda
_2$ until it reaches the constant value corresponding to $P_{e,w}$ in the interval $\lambda_2\leq\tau\leq\lambda
_1$, and then it monotonically increases in the interval $\lambda_1\leq\tau\leq\lambda
_3$ until it reaches $1$. Therefore, the constant value of the detection error rate in the interval $\lambda_2\leq\tau\leq\lambda
_1$ is the minimum value of $P_{e,w}$ for $\lambda_1\geq\lambda_2$ that can be attained using any threshold in the interval $[\lambda_2,\lambda_1]$.
\end{proof}
\noindent\textbf{Remark 1.} Eq. \eqref{pe} shows that for small values of  $P_J^{\rm max}$ with $P_J^{\rm max}G_{aw,s}|\tilde{h}_{aw,s}|^2\leq P_aG_{aw,f}|\tilde{h}_{aw,f}|^2$ Willie can attain a zero error rate negating the possibility of achieving a positive-rate covert communication as $n\to\infty$. Although increasing $P_J^{\rm max}$ beyond $P_aG_{aw,f}|\tilde{h}_{aw,f}|^2/(G_{aw,s}|\tilde{h}_{aw,s}|^2)$ can increase $P^*_{e,w}$ and enable a positive-rate covert communication ($P^*_{e,w}\to1$ as $P_J^{\rm max}\to\infty$), it also degrades the performance of the desired Alice-Bob link as we will see in Section \ref{Sec4}. The superiority of covert mmWave communication to that of omni-directional RF communication becomes then apparent by observing the beneficial impact of beamforming. In fact, in the received signal by Willie, $P_J$ is gained by $G_{aw,s}$ which is much larger than the gain $G_{aw,f}$ of $P_a$; this simultaneously increases the jamming signal and decreases the desired signal received by Willie, i.e., significantly degrades the performance of Willie's detector. It will be shown in Section \ref{Sec4} that an opposite situation happens for the Alice-Bob link where the desired signal is gained with $G_{ab,f}$ which is much larger than the gain $G_{ab,s}$ of the jamming signal. 
\begin{figure*}[b]
	\hrulefill
	\normalsize
	\setcounter{equation}{12}
	\begin{align}\label{Epew}
		\!\!\!\!\!\!\!\!	\E[P^*_{e,w}]\!=\!\!\!\!\!\!\!\sum_{\mathcal{B}\in\{{\rm L,N}\}}\!\!\!\!\!\!P_{aw}(\mathcal{B})\!\sum_{k=1}^{2}b_k^{(a,s)}\!\bigg[1\!+\!S(\nu_{\mathcal{B}},g_k^{(a,s)})\bigg]\!\!\times\!\!\bigg[1\!-\!S(\nu_{\mathcal{B}},g_k^{(a,s)})\!+\!\frac{P_am_{a,f}\nu_{\mathcal{B}}^{\nu_{\mathcal{B}}}}{P_J^{\rm max}g_k^{(a,s)}\eta_{\mathcal{B}}\Gamma(\nu_{\mathcal{B}})}\!\sum_{l=1}^{\nu_{\mathcal{B}}}\!\!\binom{\nu_{\mathcal{B}}}{l}\frac{(-1)^{l}}{l}	I(\nu_{\mathcal{B}},l,g_k^{(a,s)})\bigg]\!.\!\!
	\end{align}
\end{figure*}
\subsection{$\E[P^*_{e,w}]$ From Alice's Perspective}\label{Sec3B}
{Given that Willie is a passive node, we make the realistic assumption that Alice and Bob are unaware of the instantaneous realization of the channel between Alice and Willie. Therefore, they should rely on the expected value of $P_{e,w}^*$.}
Note also that the minimum error rate $P_{e,w}^*$ in \eqref{pe} is independent of the beamforming gain of Willie's receiver as it cancels out in the ratio of $G_{aw,f}/G_{aw,s}$ and also in  the comparison between $\lambda_1$ and $\lambda_2$. Furthermore, Alice perfectly knows the gain $m_{a,f}$ of the side lobe of her first array to Willie. However, she has uncertainty about the gain $g^{(a,s)}$ of the main lobe of the second array toward Willie due to the misalignment error; it is either $g_1^{(a,s)}\triangleq M_{a,s}$ with probability $b_1^{(a,s)}\triangleq F_{|\mathcal{E}_{a,s}|}\left({\theta_{a,s}}/{2}\right)$ or $g_2^{(a,s)}\triangleq m_{a,s}$ with probability $b_2^{(a,s)}\triangleq 1-F_{|\mathcal{E}_{a,s}|}\left({\theta_{a,s}}/{2}\right)$. Moreover, Alice and Bob do not know whether the Alice-Willie link is LOS or NLOS; hence, they should take into account two possibilities given the LOS probability $P_{\rm LOS}(d_{aw})$. In the following theorem, we characterize the expected value of $P^*_{e,w}$ form Alice's perspective in a closed form.
\begin{theorem}\label{col2}
The expected value of $P^*_{e,w}$ form Alice's perspective is characterized as \eqref{Epew}, shown at the bottom of this page,
where $P_{aw}({\rm L})\triangleq P_{\rm LOS}(d_{aw})$, $P_{aw}({\rm N})\triangleq1-P_{\rm LOS}(d_{aw})$, $\Gamma(\cdot)$ is the gamma function \cite[Eq. (8.310.1)]{gradshteyn2014table}, and $g_k^{(a,s)}$ and $b_k^{(a,s)}$ are defined above for $k\in\{1,2\}$. Moreover, the function $S(\nu_{\mathcal{B}},g_k^{(a,s)})$ is defined as
\begin{align}	\setcounter{equation}{13}
\!S(\nu_{\mathcal{B}},g_k^{(a,s)})\!\triangleq\!\sum_{l=1}^{\nu_{\mathcal{B}}}\!\binom{\nu_{\mathcal{B}}}{l}\!(-1)^l\left(\!1\!+\!l\frac{\eta_{\mathcal{B}}P_J^{\rm max}g_k^{(a,s)}}{P_am_{a,f}\nu_{\mathcal{B}}}\!\right)^{\!\!-\nu_{\mathcal{B}}}\!\!\!\!\!,
\end{align}
and
$I(\nu_{\mathcal{B}},l,g_k^{(a,s)})$, for $\nu_{\mathcal{B}}=1$ and $\nu_{\mathcal{B}}\geq2$, is defined as
\begin{align}
I(1,l,g_k^{(a,s)})&\triangleq\ln\left(1+l\frac{P_J^{\rm max}g_k^{(a,s)}}{P_am_{a,f}}\right),\label{I1}\\
I(\nu_{\mathcal{B}}\geq2,l,g_k^{(a,s)})&\triangleq\frac{(\nu_{\mathcal{B}}-2)!}{\nu_{\mathcal{B}}^{\nu_{\mathcal{B}}-1}}\bigg[1\nonumber\\
&\hspace{-0.5cm}-\left(1+l\frac{\eta_{\mathcal{B}}P_J^{\rm max}g_k^{(a,s)}}{P_am_{a,f}\nu_{\mathcal{B}}}\right)^{\!\!-\nu_{\mathcal{B}}+1}\bigg].
\end{align}
\end{theorem}
\begin{proof}
Let $P^{\rm C}_{e,w}$, $\lambda_1^{\rm C}$, and $\lambda_2^{\rm C}$ denote the values of $P^*_{e,w}$, $\lambda_1$, and $\lambda_2$, respectively, conditioned on the blockage instance $\mathcal{B}\in\{{\rm L,N}\}$ and the gain $g^{(a,s)}$ of Alice's second array to Willie. Then using \eqref{pe} we have
\begin{align}\label{Epew1}
&\hspace{-0.2cm}\E[P^{\rm C}_{e,w}]\nonumber\\
&\hspace{0.1cm}=\!\E_{\lambda_1^{\rm C}<\lambda_2^{\rm C}}[P^{\rm C}_{e,w}]\!\Pr(\lambda^{\rm C}_1\!\!<\!\!\lambda^{\rm C}_2)\!+\!\E_{\lambda^{\rm C}_1\geq\lambda^{\rm C}_2}[P^{\rm C}_{e,w}]\!\Pr(\lambda^{\rm C}_1\!\!\geq\!\!\lambda^{\rm C}_2)\!\nonumber\\
&\hspace{0.1cm}=\Pr(\lambda^{\rm C}_1\!\geq\!\lambda^{\rm C}_2)\!\left(\!1\!-\!\frac{P_am_{a,f}}{P_J^{\rm max}g^{(a,s)}}\E_{\lambda^{\rm C}_1\geq\lambda^{\rm C}_2}\!\left[\!\frac{|\tilde{h}^{(\mathcal{B})}_{aw,f}|^2}{|\tilde{h}^{(\mathcal{B})}_{aw,s}|^2}\!\right]\!\right)\!.\!
\end{align}
The closed form of $\Pr(\lambda^{\rm C}_1\geq\lambda^{\rm C}_2)$ is derived as
\begin{align}\label{eqcolproof2}
\Pr(\lambda^{\rm C}_1\!\geq\!\lambda^{\rm C}_2)&=\Pr\left(|\tilde{h}^{(\mathcal{B})}_{aw,f}|^2\leq\frac{P_J^{\rm max}g^{(a,s)}}{P_am_{a,f}}|\tilde{h}^{(\mathcal{B})}_{aw,s}|^2\right)\nonumber\\
&\hspace{-1.7cm}\stackrel{(a)}{=}\sum_{l=0}^{\nu_{\mathcal{B}}}\!\binom{\nu_{\mathcal{B}}}{l}\!(-1)^l\E_{|\tilde{h}^{(\mathcal{B})}_{aw,s}|^2}\!\!\left[\exp\!\left(\!-\eta_{\mathcal{B}}l\frac{P_J^{\rm max}g^{(a,s)}}{P_am_{a,f}}|\tilde{h}^{(\mathcal{B})}_{aw,s}|^2\!\right)\!\right]
\nonumber\\
&\hspace{-1.7cm}\stackrel{(b)}{=}\sum_{l=0}^{\nu_{\mathcal{B}}}\!\binom{\nu_{\mathcal{B}}}{l}\!(-1)^l\left(1+l\frac{\eta_{\mathcal{B}}P_J^{\rm max}g^{(a,s)}}{P_am_{a,f}\nu_{\mathcal{B}}}\right)^{\!\!-\nu_{\mathcal{B}}},
\end{align}
where step $(a)$ follows from Alzer's lemma \cite{alzer1997some}, \cite[Lemma 6]{bai2015coverage} for a normalized gamma RV $X\sim{\rm Gamma}(\nu_{\mathcal{B}},1/\nu_{\mathcal{B}})$, which states that $\Pr\left(X<x\right)$ is tightly approximated by $\left[1-\exp(-\eta_{\mathcal{B}} x)\right]^{\nu_{\mathcal{B}}}$ where $\eta_{{\mathcal{B}}}=\nu_{\mathcal{B}}(\nu_{\mathcal{B}}!)^{-1/\nu_{\mathcal{B}}}$, and then applying the binomial theorem assuming $\nu_{\mathcal{B}}$ is an integer \cite{bai2015coverage}, i.e., 
\begin{align}\label{FXx}
F_X(x)=\sum_{l=0}^{\nu_{\mathcal{B}}}\binom{\nu_{\mathcal{B}}}{l}(-1)^{l}{\rm e}^{-l\eta_{\mathcal{B}}x}.    
\end{align}
Moreover, step $(b)$ is derived using the moment generating function (MGF) of a normalized gamma RV $X$, i.e., $\E[{\rm e}^{tX}]=(1-t/\nu_{\mathcal{B}})^{-\nu_{\mathcal{B}}}$ for any $t<\nu_{\mathcal{B}}$.

Moreover, for the expectation term in \eqref{Epew1} we have
\begin{align}\label{eqcolproof3}
&\hspace{-0.2cm}\E_{\lambda^{\rm C}_1\geq\lambda^{\rm C}_2}\!\left[\frac{|\tilde{h}^{(\mathcal{B})}_{aw,f}|^2}{|\tilde{h}^{(\mathcal{B})}_{aw,s}|^2}\right]\nonumber\\
&\hspace{0.5cm}=\E\!\left[\frac{|\tilde{h}^{(\mathcal{B})}_{aw,f}|^2}{|\tilde{h}^{(\mathcal{B})}_{aw,s}|^2}{\Bigg|}|\tilde{h}^{(\mathcal{B})}_{aw,f}|^2\leq\frac{P_J^{\rm max}g^{(a,s)}}{P_am_{a,f}}|\tilde{h}^{(\mathcal{B})}_{aw,s}|^2\right]\nonumber\\
&\hspace{0.5cm}=\int_{0}^{\infty}\frac{f_{|\tilde{h}^{(\mathcal{B})}_{aw,s}|^2}(y)}{y}\Bigg[\underbrace{\int_{0}^{C_1y}xf_{|\tilde{h}^{(\mathcal{B})}_{aw,f}|^2}(x)dx}_{V_1}\Bigg]dy,
\end{align}
where $C_1\triangleq\frac{P_J^{\rm max}g^{(a,s)}}{P_am_{a,f}}$, and $f_{|\tilde{h}^{(\mathcal{B})}_{aw,f}|^2}(x)$ and $f_{|\tilde{h}^{(\mathcal{B})}_{aw,s}|^2}(y)$ are the PDFs of the fading coefficients $|\tilde{h}^{(\mathcal{B})}_{aw,f}|^2$ and $|\tilde{h}^{(\mathcal{B})}_{aw,s}|^2$, respectively. Applying the part-by-part integration rule to $V_1$ and then using Alzer's lemma together with the binomial theorem as \eqref{FXx} yields
\begin{align}\label{eqcolproof4}
V_1=&C_1y\sum_{l_1=0}^{\nu_{\mathcal{B}}}\!\binom{\nu_{\mathcal{B}}}{l_1}\!(-1)^{l_1}{\rm e}^{-l_1\eta_{\mathcal{B}}C_1y}\nonumber\\
&-C_1y-\sum_{l_2=1}^{\nu_{\mathcal{B}}}\!\binom{\nu_{\mathcal{B}}}{l_2}\frac{(-1)^{l_2}}{\eta_{\mathcal{B}}l_2}\left[1-{\rm e}^{-l_2\eta_{\mathcal{B}}C_1y}\right].
\end{align}
By plugging \eqref{eqcolproof4} into \eqref{eqcolproof3}, using the MGF of the normalized gamma RV $|\tilde{h}^{(\mathcal{B})}_{aw,s}|^2$, and then noting that $f_{|\tilde{h}^{(\mathcal{B})}_{aw,s}|^2}(y)=\nu_{\mathcal{B}}^{\nu_{\mathcal{B}}}y^{\nu_{\mathcal{B}}-1}{\rm e}^{-\nu_{\mathcal{B}}y}/\Gamma(\nu_{\mathcal{B}})$ we have
\begin{align}\label{eqcolproof5}
&\!\!\!\E_{\lambda^{\rm C}_1\geq\lambda^{\rm C}_2}\!\left[\!\frac{|\tilde{h}^{(\mathcal{B})}_{aw,f}|^2}{|\tilde{h}^{(\mathcal{B})}_{aw,s}|^2}\!\right]\!=\!C_1\!\sum_{l_1=1}^{\nu_{\mathcal{B}}}\!\binom{\nu_{\mathcal{B}}}{l_1}\!(-1)^{l_1}\!\left(\!1\!+\!l_1\frac{\eta_{\mathcal{B}}C_1}{\nu_{\mathcal{B}}}\right)^{\!\!-\nu_{\mathcal{B}}}\nonumber\\
&-\sum_{l_2=1}^{\nu_{\mathcal{B}}}\!\binom{\nu_{\mathcal{B}}}{l_2}\frac{(-1)^{l_2}\nu_{\mathcal{B}}^{\nu_{\mathcal{B}}}}{\eta_{\mathcal{B}}l_2\Gamma(\nu_{\mathcal{B}})}\bigg[\int_{0}^{\infty}{y^{\nu_{\mathcal{B}}-2}}{\rm e}^{-\nu_{\mathcal{B}}y}dy\nonumber\\
&\hspace{3cm}-\int_{0}^{\infty}{y^{\nu_{\mathcal{B}}-2}}{\rm e}^{-(l_2\eta_{\mathcal{B}}C_1+\nu_{\mathcal{B}})y}dy\bigg].
\end{align}
Now given that the parameter $\nu_{\mathcal{B}}$ of Nakagami-$m$ fading is always greater than or equal  to $0.5$ and is assumed to be an integer here, we have $\nu_{\mathcal{B}}\in\mathbb{N}$ where $\mathbb{N}$  stands for the set of natural numbers.
For $\nu_{\mathcal{B}}\geq2$, by \cite[Eq. (3.351.3)]{gradshteyn2014table} we have $\int_{0}^{\infty}y^{\nu_{\mathcal{B}}-2}{\rm e}^{-\alpha y}dy=(\nu_{\mathcal{B}}-2)!/\alpha^{\nu_{\mathcal{B}}-1}$ for any $\alpha \in \R^+$. On the other hand, for $\nu_{\mathcal{B}}=1$ using \cite[Eq. (2.325.1)]{gradshteyn2014table} we have $\int_{0}^{\infty}y^{-1}{\rm e}^{-\alpha y}dy={\rm Ei}(-\alpha y)|_0^{\infty}$, {where ${\rm Ei}(\cdot)$ is the exponential integral function defined as \cite[Eq. (8.211.1)]{gradshteyn2014table} for negative arguments}. Therefore, following a similar approach to the proof of \cite[Corollary 2]{jamali2019uplink} we can calculate the difference of the two integrals in \eqref{eqcolproof5} as $\lim_{y\to0}[{\rm Ei}(-(l_2\eta_{\mathcal{B}}C_1+\nu_{\mathcal{B}})y)-{\rm Ei}(-\nu_{\mathcal{B}}y)]=\ln([l_2\eta_{\mathcal{B}}C_1+\nu_{\mathcal{B}}]/\nu_{\mathcal{B}})$ which is equal to $\ln(1+l_2C_1)$ for $\nu_{\mathcal{B}}=1$ (note that $\eta_{\mathcal{B}}=1$ for $\nu_{\mathcal{B}}=1$, and ${\rm Ei}(-\infty)=0$). This completes the proof of the theorem given  the definition of $I(\nu_{\mathcal{B}},l,g^{(a,s)})$ in \Tref{col2}. 
\end{proof}

\noindent\textbf{Remark 2.} In \Tref{col2}, it is assumed that Willie is not in the main lobe of Alice's first antenna array and hence, receives the desired signal by a side lobe gain $m_{a,f}$. However, if Willie is within the main lobe of the first array, we should include another averaging over the gain $g^{(a,f)}$ of the first array given the beamsteering error, i.e., that gain is either $g_1^{(a,f)}\triangleq M_{a,f}$ with probability $b_1^{(a,f)}\triangleq F_{|\mathcal{E}_{a,f}|}\left({\theta_{a,f}}/{2}\right)$ or $g_2^{(a,f)}\triangleq m_{a,f}$ with probability $b_2^{(a,f)}\triangleq 1-F_{|\mathcal{E}_{a,f}|}\left({\theta_{a,f}}/{2}\right)$.

\section{Performance of the  Alice-Bob Link}\label{Sec4}
In this section, we characterize performance metrics of the Alice-Bob link including its outage probability, maximum effective covert rate (i.e., the rate for which Alice can reliably communicate with Bob while maintaining $\E[P^*_{e,w}]\geq 1-\epsilon$ for any given $\epsilon>0)$, and ergodic capacity.
\subsection{Outage Probability}\label{Sec4A}
We assume that Alice targets a rate $R_b$ requiring the Alice-Bob link to meet a threshold signal-to-interference-plus-noise ratio (SINR) $\gamma_{\rm th}\triangleq2^{R_b}-1$. Then the outage probability $P_{\rm out}^{\rm AB}\triangleq\Pr({\gamma_{ab}}<\gamma_{\rm th})$ in achieving $R_b$ is characterized, in a closed form, in \Tref{thm3}, where the SINR $\gamma_{ab}$ of the  Alice-Bob link is given as follows by using \eqref{yb}:
\begin{align}\label{gab}
\gamma_{ab}=\frac{{P_aG_{ab,f}L_{ab}}|\tilde{h}_{ab,f}|^2}{{P_JG_{ab,s}L_{ab}}|\tilde{h}_{ab,s}|^2+\sigma^2_b}.
\end{align}
Note that in addition to $|\tilde{h}_{ab,f}|^2$, $|\tilde{h}_{ab,s}|^2$, and $P_J$, the blockage instance $\mathcal{B}\in\{{\rm L},{\rm N}\}$ and the antenna gains can also change randomly across transmission blocks. In particular, while we assume that the jamming signal arrives with the deterministic side lobe gain $m_{a,s}$, there are still uncertainties in the gains of Alice's first array and Bob's receiver (they are pointing their main lobes) due to the beamsteering error. Therefore, the gain $g^{(a,f)}$ of the main lobe of Alice's first array pointed to Bob is either $g_1^{(a,f)}\triangleq M_{a,f}$ with probability $b_1^{(a,f)}\triangleq F_{|\mathcal{E}_{a,f}|}\left({\theta_{a,f}}/{2}\right)$ or $g_2^{(a,f)}\triangleq m_{a,f}$ with probability $b_2^{(a,f)}\triangleq 1-F_{|\mathcal{E}_{a,f}|}\left({\theta_{a,f}}/{2}\right)$. Similarly, the gain $g^{(b)}$ of Bob's receiver is either $g_1^{(b)}\triangleq M_{b}$ with probability $b_1^{(b)}\triangleq F_{|\mathcal{E}_{b}|}\left({\theta_{b}}/{2}\right)$ or $g_2^{(b)}\triangleq m_{b}$ with probability $b_2^{(b)}\triangleq 1-F_{|\mathcal{E}_{b}|}\left({\theta_{b}}/{2}\right)$. Furthermore, in \Tref{thm3} we assume that Willie is not in the main lobe of Alice's first array. However, if Willie is in the Alice-Bob direction, we should include another averaging of the gain of Alice's second array carrying the jammer signal, i.e., instead of a deterministic $m_{a,s}$ we should consider two possibilities $g_k^{(a,s)}$ with probabilities $b_k^{(a,s)}$, $k\in\{1,2\}$, defined in Section III-B.

\begin{theorem}\label{thm3}
The outage probability of the Alice-Bob link in achieving the target rate $R_b\triangleq\log_2(1+\gamma_{\rm th})$ is given by
\begin{align}\label{pout}
&P_{\rm out}^{\rm AB}\!=\!\!\!\!\sum_{\mathcal{B}\in\{{\rm L,N}\}}\!\!\!\!P_{ab}(\mathcal{B})\sum_{k_1=1}^{2}b_{k_1}^{(a,f)}\sum_{k_2=1}^{2}b_{k_2}^{(b)}\bigg[1\!+\!\sum_{l=1}^{\nu_{\mathcal{B}}}\!\binom{\nu_{\mathcal{B}}}{l}(-1)^{l}\nonumber\\
&\hspace{1cm}\times\exp\!\bigg(\!-\frac{l\eta_{\mathcal{B}}\gamma_{\rm th}\sigma^2_b}{P_ag_{k_1}^{(a,f)}g^{(b)}_{k_2}L_{ab}^{(\mathcal{B})}}\bigg)V(\nu_{\mathcal{B}},l,g^{(a,f)}_{k_1})\bigg],
\end{align}
where $P_{ab}({\rm L})\triangleq P_{\rm LOS}(d_{ab})$ and $P_{ab}({\rm N})\triangleq 1-P_{\rm LOS}(d_{ab})$. Also,
$V(\nu_{\mathcal{B}},l,g^{(a,f)}_{k_1})$, for $\nu_{\mathcal{B}}=1$ and $\nu_{\mathcal{B}}\geq2$, is defined as
\begin{align}
&\hspace{-2.8cm}\!V(1,l,g^{(a,f)}_{k_1})\!\triangleq\!\frac{P_ag^{(a,f)}_{k_1}}{P_J^{\rm max}l\gamma_{\rm th}m_{a,s}}\ln\!\bigg(\!1\!+\!\frac{P_J^{\rm max}l\gamma_{\rm th}m_{a,s}}{P_ag^{(a,f)}_{k_1}}\!\bigg),\!\\
V(\nu_{\mathcal{B}}\geq2,l,g^{(a,f)}_{k_1})&\triangleq\frac{\nu_{\mathcal{B}}P_ag^{(a,f)}_{k_1}}{P_J^{\rm max}l\eta_{\mathcal{B}}\gamma_{\rm th}m_{a,s}(\nu_{\mathcal{B}}-1)}\nonumber\\
&\hspace{-0.7cm}\times\bigg[1-\bigg(1+\frac{P_J^{\rm max}l\eta_{\mathcal{B}}\gamma_{\rm th}m_{a,s}}{\nu_{\mathcal{B}}P_ag^{(a,f)}_{k_1}}\bigg)^{\!\!1-\nu_{\mathcal{B}}}\bigg].
\end{align}
\end{theorem}
\begin{proof}
Given the SINR of the Alice-Bob link in \eqref{gab}, the outage probability conditioned on the blockage instance $\mathcal{B}$ as well as the antenna gains $g^{(a,f)}$ and $g^{(b)}$ is characterized as follows:
\begin{align}\label{proofthm3-1}
P_{\rm out,C}^{\rm AB}\triangleq&\Pr(\gamma_{ab}<\gamma_{\rm th}|\mathcal{B},g^{(a,f)},g^{(b)})\nonumber\\
\stackrel{(a)}{=}&\Pr\left(|\tilde{h}_{ab,f}^{(\mathcal{B})}|^2<C_2P_J|\tilde{h}_{ab,s}^{(\mathcal{B})}|^2+C_3\right)\nonumber\\
&\hspace{-1.4cm}\stackrel{(b)}{=}\sum_{l=0}^{\nu_{\mathcal{B}}}\binom{\nu_{\mathcal{B}}}{l}(-1)^{l}{\rm e}^{-l\eta_{\mathcal{B}}C_3}\E_{P_J,|\tilde{h}^{(\mathcal{B})}_{ab,s}|^2}\!\!\left[{\rm e}^{-l\eta_{\mathcal{B}}C_2P_J|\tilde{h}^{(\mathcal{B})}_{ab,s}|^2}\!\right]\nonumber\\
&\hspace{-1.4cm}\stackrel{(c)}{=}\sum_{l=0}^{\nu_{\mathcal{B}}}\binom{\nu_{\mathcal{B}}}{l}(-1)^{l}{\rm e}^{-l\eta_{\mathcal{B}}C_3}\E_{P_J}\!\!\left[\!\left(1+\frac{l\eta_{\mathcal{B}}C_2P_J}{\nu_{\mathcal{B}}}\right)^{\!\!-\nu_{\mathcal{B}}}\right]
\nonumber\\
&\hspace{-1.4cm}\stackrel{(d)}{=}\!1\!+\!\sum_{l=1}^{\nu_{\mathcal{B}}}\!\binom{\nu_{\mathcal{B}}}{l}\!\frac{(-1)^{l}{\rm e}^{-l\eta_{\mathcal{B}}C_3}}{P_J^{\rm max}}\!\!\int_{0}^{P_J^{\rm max}}\!\!\!\!\!\!\!\Big(\!1\!+\!\frac{l\eta_{\mathcal{B}}C_2x}{\nu_{\mathcal{B}}}\Big)^{\!\!-\nu_{\mathcal{B}}}\!\!dx,\!\!
\end{align}
where in step $(a)$ we have defined $C_2\triangleq \gamma_{\rm th}m_{a,s}/(P_ag^{(a,f)})$ and $C_3\triangleq \gamma_{\rm th}\sigma^2_b/(P_ag^{(a,f)}g^{(b)}L_{ab}^{(\mathcal{B})})$. Moreover, step $(b)$ follows by Alzer's lemma together with the binomial theorem as \eqref{FXx}, and step $(c)$ is derived using the MGF of the normalized gamma RV $|\tilde{h}^{(\mathcal{B})}_{ab,s}|^2$. Finally,  taking the integral  in step $(d)$ and recalling the definition of  the function $V(\nu_{\mathcal{B}},l,g^{(a,f)}_{k_1})$ from the statement of the theorem complete the proof. 
\end{proof}
\subsection{Maximum Effective Covert Rate}\label{Sec4B}
Given any target data rate $R_b$, Alice and Bob can have the effective communication rate $\overline{R}_{a,b}\triangleq R_b(1-P_{\rm out}^{\rm AB})$, where their outage probability $P_{\rm out}^{\rm AB}$ in achieving the target rate $R_b$ is obtained using \Tref{thm3}. The goal here is to determine the optimal value of $P_J^{\rm max}$ that maximizes $\overline{R}_{a,b}$ while also satisfying the covertness requirement, i.e., $\E[P^*_{e,w}]\geq 1-\epsilon$ for any given $\epsilon>0$. We first note that $\E[P^*_{e,w}]$ and $P_{\rm out}^{\rm AB}$  both monotonically increase with $P_J^{\rm max}$ (see also Fig. \ref{Fig1} and Fig. \ref{Fig2} for the visualization). Then, in order to obtain the maximum effective covert rate $\overline{R}^*_{a,b}$ achievable in our setup, we need to pick the smallest possible value for $P_J^{\rm max}$ given that $\overline{R}_{a,b}$ is monotonically decreasing with respect to $P_J^{\rm max}$. This smallest possible value for $P_J^{\rm max}$, denoted by $P_{J,\rm opt}^{\rm max}$, must also satisfy the covertness requirement $\E[P^*_{e,w}]\geq 1-\epsilon$ for the given $\epsilon>0$. Now, given that $\E[P^*_{e,w}]$ monotonically increases with $P_J^{\rm max}$, the solution of the equation $\E[P^*_{e,w}]=1-\epsilon$ for $P_J^{\rm max}$ defines $P_{J,\rm opt}^{\rm max}$. This observation is summarized in the following proposition. Note, however, that the optimal rate per \Pref{prop4} needs to be evaluated numerically. 

\begin{proposition}\label{prop4}
Given fixed system and channel parameters, fixed covertness requirement $\epsilon$, and target data rate $R_b$, the maximum effective covert rate achievable in the considered setup is equal to $R_b(1-P_{\rm out}^{*\rm AB})$, denoted by $\overline{R}^*_{a,b}$, {where $P_{\rm out}^{*\rm AB}$ is equal to $P_{\rm out}^{\rm AB}$, specified in \eqref{pout}, evaluated in $P_{J,\rm opt}^{\rm max}$ that is the solution of the equation $\E[P^*_{e,w}]=1-\epsilon$ for $P_J^{\rm max}$.}
\end{proposition}

\subsection{Ergodic Capacity}\label{Sec4C}
In addition to characterizing the maximum effective covert rate $\overline{R}^*_{a,b}$ given a target rate $R_b$, provided in Section \ref{Sec4B}, it is desirable to determine the achievable average data rate of the Alice-Bob link, referred to as its \textit{ergodic capacity}, given fixed values for the parameters involved in the model\footnote{We assume that the set of parameters is chosen such that the covert communication requirement $\E[P^*_{e,w}]\geq 1-\epsilon$ is satisfied for any $\epsilon>0$. Note that $\E[P^*_{e,w}]$ depends only on the values of the design parameters as well as the statistics of the RVs involved and not on their instantaneous realizations.}. The ergodic capacity $\E[R_{a,b}]$ of the Alice-Bob link is obtained while assuming that the threshold/target data rate $R_b$ is adjusted by the channel conditions, i.e., $\gamma_{\rm th}=\gamma_{ab}$, implying that Bob's decoder can always decode the received signal without outage. In fact, given the instantaneous SINR $\gamma_{ab}$, specified in \eqref{gab}, Alice can reliably transmit to Bob with the data rate equal to $\log_2(1+\gamma_{ab})$. Therefore, on average, the data rate $\E[R_{a,b}]\triangleq \E[\log_2(1+\gamma_{ab})]$ is achievable for the Alice-Bob link, where the expectation is over the RVs involved in \eqref{gab}. In the following theorem, we characterize $\E[R_{a,b}]$ in a tractable form that involves only one-dimensional integrals over one of the fading coefficients. 
\begin{theorem}\label{Thm5}
The ergodic capacity $\E[R_{a,b}]$ of the Alice-Bob link is given by
\begin{align}\label{eq4C1}
\!\E[R_{a,b}]=&\frac{P_a}{m_{a,s}P_J^{\rm max}\ln 2}\sum_{\mathcal{B}\in\{{\rm L,N}\}}\!\!\!P_{ab}(\mathcal{B})\sum_{k_1=1}^{2}\!b_{k_1}^{(a,f)}g_{k_1}^{(a,f)}\nonumber\\
&\times\sum_{k_2=1}^{2}\!b_{k_2}^{(b)}\sum_{l=1}^{\nu_{\mathcal{B}}}\binom{\nu_{\mathcal{B}}}{l}\frac{(-1)^{l}}{l\eta_{\mathcal{B}}}\left[\mathcal{J}_1-\mathcal{J}_2-\mathcal{J}_3\right],
\end{align}
where $\mathcal{J}_1$, $\mathcal{J}_2$, and $\mathcal{J}_3$ are defined in the form of one-dimensional integrals as follows:
\begin{align}
\mathcal{J}_1&\triangleq\int_{0}^{\infty}\frac{1}{y}\Bigg[{\rm eEi}\Bigg(\frac{l\eta_{\mathcal{B}}}{P_ag_{k_1}^{(a,f)}}\Bigg[m_{a,s}P_J^{\rm max}y\nonumber\\
&\hspace{3.5cm}+\frac{\sigma^2_b}{g_{k_2}^{(b)}L_{ab}^{(\mathcal{B})}}\Bigg]\Bigg)\Bigg]f_Y(y)dy,\label{J1}\\
\mathcal{J}_2&\triangleq\left[{\rm eEi}\!\left(\frac{l\eta_{\mathcal{B}}\sigma^2_b}{P_ag_{k_1}^{(a,f)}g_{k_2}^{(b)}L_{ab}^{(\mathcal{B})}}\right)\right]\int_{0}^{\infty}\frac{1}{y}f_Y(y)dy,\label{J2}\\
\mathcal{J}_3&\triangleq\int_{0}^{\infty}\frac{1}{y}\!\left[\ln\!\left(\!1\!+\!\frac{m_{a,s}g_{k_2}^{(b)}L_{ab}^{(\mathcal{B})}P_J^{\rm max}}{\sigma^2_b}y\right)\right]\!f_Y(y)dy\label{J3},
\end{align}
where ${\rm eEi}(x)\triangleq {\rm e}^x{\rm Ei}(-x)$, 
and  $f_Y(y)$ is the PDF of a normalized gamma RV as in \eqref{pdf_gamma}.
\end{theorem}
\begin{proof}
Based on the system model considered in this paper and our earlier discussions, we have
\begin{align}\label{eqthm5p1}
\!\E[R_{a,b}]\!=\!\!\!\!\!\!\sum_{\mathcal{B}\in\{{\rm L,N}\}}\!\!\!\!\!\!P_{ab}(\mathcal{B})\!\!\sum_{k_1=1}^{2}\!b_{k_1}^{(a,f)}\!\!\sum_{k_2=1}^{2}\!b_{k_2}^{(b)}\E[R_{a,b}|\mathcal{B},g^{(a,f)}\!,g^{(b)}],
\end{align}
where $\E[R_{a,b}|\mathcal{B},g^{(a,f)}\!,g^{(b)}]$ is the ergodic capacity conditioned on the blockage instance $\mathcal{B}$, and the antenna gains $g^{(a,f)}$ and $g^{(b)}$. Given  the definition of the ergodic capacity $\E[R_{a,b}]\triangleq \E[\log_2(1+\gamma_{ab})]$ and the expression of the SINR $\gamma_{ab}$ in \eqref{gab}, we have
\begin{align}\label{eqthm5p2}
\!\E[R_{a,b}|\mathcal{B},g^{(a,f)}\!,g^{(b)}]\!=\!\E_{Y,P_J}\!\!\Bigg[{\int_{0}^{\infty}\log_2(1+C'_1x)f_X(x)dx}\Bigg],
\end{align}
where $X\triangleq |\tilde{h}_{ab,f}^{(\mathcal{B})}|$ and $Y\triangleq |\tilde{h}_{ab,s}^{(\mathcal{B})}|$ represent the RVs associated with the involved fading coefficients with PDFs $f_X(x)$ and $f_Y(y)$, respectively. Moreover, $C'_1\triangleq 1/(C'_2P_JY+C'_3)$ with $C'_2\triangleq m_{a,s}/(P_ag^{(a,f)})$ and $C'_3\triangleq \sigma^2_b/(P_ag^{(a,f)}g^{(b)}L_{ab}^{(\mathcal{B})})$. Observe that for a given RV $Z$ with the PDF $f_Z(z)$ and CDF $F_Z(z)$ we have the following part-by-part integration equality
\begin{align}\label{eqthm5p3}
&\int_{a}^{b}\log_2(1+cz)f_Z(z)dz=\frac{1}{\ln 2}\Bigg[c\int_{a}^{b}\frac{1-F_Z(z)}{1+cz}dz\nonumber\\
&+(1-F_Z(a))\ln(1+ca)-(1-F_Z(b))\ln(1+cb)\Bigg],
\end{align}
with $c$ being a constant. Then the integral involved in \eqref{eqthm5p2} is computed as follows:
\begin{align}\label{eqthm5p4}
&\int_{0}^{\infty}\log_2(1+C'_1x)f_X(x)dx\stackrel{(a)}{=}\frac{C'_1}{\ln 2}\int_{0}^{\infty}\frac{1-F_X(x)}{1+C'_1x}dx
\nonumber\\
&\hspace{1cm}\stackrel{(b)}{=}\frac{1}{\ln 2}\sum_{l=1}^{\nu_{\mathcal{B}}}\binom{\nu_{\mathcal{B}}}{l}(-1)^{l}{\rm e}^{l\eta_{\mathcal{B}}/C'_1}{\rm Ei}({-l\eta_{\mathcal{B}}/C'_1}),
\end{align}
where step $(a)$ is by using \eqref{eqthm5p3}, and noting that $(1-F_X(0))\ln(1+C'_1\times0)=0
$ and
\begin{align}
\lim_{x\to\infty}(1-F_X(x))\ln(1+C'_1x)=0,
\end{align}
since $1-F_X(x)=-\sum_{l=1}^{\nu_{\mathcal{B}}}\binom{\nu_{\mathcal{B}}}{l}(-1)^{l}{\rm e}^{-l\eta_{\mathcal{B}}x}$ decays exponentially while $\ln(1+C'_1x)$ grows logarithmically with $x$. Moreover, step $(b)$ is derived by first using Alzer's lemma together with the binomial theorem as \eqref{FXx},
and then applying \cite[Eq. (3.352.4)]{gradshteyn2014table}. Now by substituting \eqref{eqthm5p4} into \eqref{eqthm5p2}, we have for the conditional ergodic capacity as
\begin{align}\label{eqthm5p5}
&\E[R_{a,b}|\mathcal{B},g^{(a,f)}\!,g^{(b)}]=\frac{1}{P_J^{\rm max}\ln 2}\sum_{l=1}^{\nu_{\mathcal{B}}}\binom{\nu_{\mathcal{B}}}{l}(-1)^{l}\E_{Y}\!\!\Bigg[\nonumber\\
&\underbrace{{\int_{0}^{P_J^{\rm max}}\!\!\!\!\!\!\exp\!\left(l\eta_{\mathcal{B}}(C'_2Yt\!+\!C'_3)\right){\rm Ei}\!\left(-l\eta_{\mathcal{B}}(C'_2Yt\!+\!C'_3)\right)dt}}_{I_1}\Bigg]\!.\!
\end{align}
The integral term $I_1$ can be computed in a closed form as
\begin{align}\label{eqthm5p6}
I_1\stackrel{(a)}{=}&\frac{1}{C'_2Y}\int_{C'_3}^{C'_2YP_J^{\rm max}+C'_3}\exp\!\left(l\eta_{\mathcal{B}}z\right){\rm Ei}\!\left(-l\eta_{\mathcal{B}}z\right)dz\nonumber\\
\stackrel{(b)}{=}&\frac{1}{l\eta_{\mathcal{B}}C'_2Y}\bigg[{\rm eEi}\!\left(l\eta_{\mathcal{B}}(C'_2YP_J^{\rm max}+C'_3)\right)\nonumber\\
&-{\rm eEi}\!\left(l\eta_{\mathcal{B}}C'_3\right)-\ln\!\left(1+C'_2YP_J^{\rm max}/C'_3\right)\bigg],
\end{align}
where step $(a)$ is by defining $z\triangleq C'_2Yt\!+\!C'_3$, and step $(b)$ is obtained using \Lref{lemma_app1} in Appendix \ref{App_A} and defining ${\rm eEi}(x)\triangleq {\rm e}^x{\rm Ei}(-x)$. Finally, taking the expectation of $I_1$ in \eqref{eqthm5p6} over $Y$ and then plugging \eqref{eqthm5p5} back into \eqref{eqthm5p1} complete the proof.
\end{proof}

To the best of our knowledge, the one-dimensional integrals in \eqref{J1}-\eqref{J3} cannot be computed in closed forms for all values of $\nu_{\mathcal{B}}$. In particular, obtaining closed-form expressions for the special case  $\nu_{\mathcal{B}}=1$, which corresponds to Rayleigh fading channels, is not straightforward mainly due to the fractional term $1/y$ in the integrands. On the other hand, as delineated in the following proposition, deriving the closed forms of \eqref{J1}-\eqref{J3} for the special case $\nu_{\mathcal{B}}=2$ is straightforward.
\begin{proposition}\label{Prop6}
For $\nu_{\mathcal{B}}=2$, the closed-form expressions for $\mathcal{J}_1$, $\mathcal{J}_2$, and $\mathcal{J}_3$, defined in \Tref{Thm5}, are as follows:
\begin{align}
\mathcal{J}_1&=\frac{4P_ag_{k_1}^{(a,f)}}{l\eta_{\mathcal{B}}m_{a,s}P_J^{\rm max}\!-\!2P_ag_{k_1}^{(a,f)}}\Bigg[{\rm eEi}\left(\!\frac{2\sigma^2_b}{g_{k_2}^{(b)}L_{ab}^{(\mathcal{B})}m_{a,s}P_J^{\rm max}}\!\right)\nonumber\\
&\hspace{2cm}-{\rm eEi}\left(\frac{l\eta_{\mathcal{B}}\sigma^2_b}{P_ag_{k_1}^{(a,f)}g_{k_2}^{(b)}L_{ab}^{(\mathcal{B})}}\right)\Bigg],\\
\mathcal{J}_2&=2~\!{\rm eEi}\!\left(\frac{l\eta_{\mathcal{B}}\sigma^2_b}{P_ag_{k_1}^{(a,f)}g_{k_2}^{(b)}L_{ab}^{(\mathcal{B})}}\right),\label{J2prop}\\
\mathcal{J}_3&=-2~\!{\rm eEi}\!\left(\frac{2\sigma^2_b}{m_{a,s}g_{k_2}^{(b)}L_{ab}^{(\mathcal{B})}P_J^{\rm max}}\right).
\end{align}
\end{proposition}
\begin{proof} The proof follows by substituting $f_Y(y)=4y{\rm e}^{-2y}$. Then the closed forms for $\mathcal{J}_1$ and $\mathcal{J}_3$ are derived by applying \cite[Corollary 1]{jamali2019uplink} and \cite[Eq. (4.337.2)]{gradshteyn2014table}, respectively.
\end{proof}
It is worth mentioning at the end that rather complicated closed forms can also be obtained for the case of $\nu_{\mathcal{B}}>2$ by employing, e.g., \cite[Eq. (06.35.21.0016.01)]{wolfram}, \cite[Eq. (3.351.3)]{gradshteyn2014table}, and \cite[Eq. (4.358.1)]{gradshteyn2014table} to solve the integrals involved in \eqref{J1}, \eqref{J2}, and \eqref{J3}, respectively.

\section{Practical Scenarios, Discussions, and Future Directions}\label{Sec5}
In this section, we first describe the localization issue in covert mmWave communications and propose a potential design approach that can be incorporated in the context of the system model in this paper. We then establish how the performance metrics of the proposed scheme can be characterized using the earlier results in this paper. Finally, we highlight several interesting future research directions. 
\subsection{Localization Issue}\label{Sec5A}
One of the important aspects in the design of a typical mmWave communication network is the localization of the nodes.
This is mainly due to the highly directive beams used in mmWave communication systems.
In the context of the system model in this paper, while it is important in the design of the system to know both Bob's and Willie's locations, obtaining the information about Willie's location is ought to be much more challenging. In fact, the legitimate parties Alice and Bob can apply sophisticated beam training approaches to establish a directional link. However, since Willie is a passive node, it is more difficult for Alice to obtain precise information about Willie's location. 

Speaking of Willie's location, both Alice's distance to Willie and the spatial direction between them are important to set up the covert mmWave communication system. However, the distance between Alice and Willie is less challenging if the spatial direction between them is known or if the direction issue is properly addressed in the system design. In fact, all of our earlier derivations are in terms of the link length $d_{aw}$ between Alice and Willie. Therefore, the performance metrics change with respect to the distance between Alice and Willie, and one needs to adopt a new set of values for the involved parameters to ensure the covertness requirement while, in the mean time, maximizing the effective rate between Alice and Bob. Note that the uncertainty about Alice's distance to Willie also exists in conventional RF-based covert communication systems incorporating omni-directional antennas and is not particular to the case of covert mmWave communication. 

On the other hand, the uncertainty about the spatial direction between Alice and Willie is more challenging as it directly impacts the design architecture for Alice's transmitter. Throughout the paper we assumed that Willie's direction is known to Alice such that the main lobe of Alice's second antenna array, carrying the jamming signal, is pointed toward Willie. However, in the case of uncertainty about Willie's direction we might not be able to do that; as a result, the jamming signal may arrive to Willie with a much lower gain of the side lobe instead of the main lobe of the second array. This deteriorates the system performance by improving Willie's detection performance which in turn degrades the Alice-Bob link performance by, e.g., requiring Alice to employ lower signal powers $P_a$ or larger jammer powers $P_J^{\rm max}$ to satisfy the covertness requirement. 

One immediate solution to address the aforementioned issue on the uncertainty about Willie's direction is to employ an antenna array with a wide (main lobe) beamwidth to transmit the jamming signal. However, it is very difficult to cover the whole space (except the Alice-Bob direction) using a single wide main lobe \cite{zhang20165g}. Therefore, Alice may prefer to employ several wide-beam antenna arrays to transmit the jamming signal. In this case, the main lobe of each array covers a certain spatial direction such that the union of the main lobes covers the whole space except the Alice-Bob direction. As a result, from the design perspective, we no longer need to know Willie's direction as the whole space is covered by multiple antenna arrays leaving a null space (or negligible side lobes) toward Bob's direction. However, from the analysis point of view, it is not easy to derive tractable forms for the system performance metrics as discussed next. In fact, assuming the sectored-pattern antenna model, as in Section \ref{Sec2A}, each antenna array has a main lobe and a side lobe. Therefore, the jammer signal arrives at Willie by a main lobe from the antenna array covering the Alice-Willie direction in addition to several side lobes each from all other antenna arrays. Similarly, the jammer signal arrives at Bob by the side lobe of all arrays other than the first antenna array. Given the spatial distance between the antenna arrays, the channel between each side lobe and the receiver, either Willie or Bob, has to be assigned an independent fading coefficient. Therefore, the received signals by Willie and Bob involve several independent Nakagami fading coefficients making it difficult to derive tractable forms for the performance metrics. 
In the next subsection, we elucidate how the results in the paper can be applied to approximate the performance metrics with respect to this multi-array system model that resolves the issue of uncertainty about Alice-Willie direction.

\noindent\textbf{Remark 3.} The system model considered in this paper and the subsequent analytical results are directly applicable to the case where an external jammer node with a single or multiple antenna arrays is used to transmit the jamming signal instead of (an) extra antenna array(s) in Alice.

\subsection{Approximate Performance of the Multi-Array Transmitter}\label{Sec5B}
Consider a system model same as the one described in Section \ref{Sec2B} except that Alice is equipped with $N_J$ wide-beam arrays, instead of one, each carrying the same jamming signal and together covering the whole space except the Alice-Bob direction. 
The main lobe gain of the antenna arrays carrying the jamming signals is assumed to be the same, denoted by $M_{a,s}$.
Now, we make the following two assumptions to approximate the system performance using tractable forms.

\textit{1- Zero side lobe gains from jamming arrays to Willie:} Note that, regardless of Willie's location, the main portion of jamming signal reaches Willie by a main lob gain $M_{a,s}$ from one of $N_J$ arrays, denoted by the $j_1$-th array, that is covering Willie's spatial direction. Then the received jamming signal by Willie at the $i$-th channel use is expressed as
\begin{align}\label{ywj}
\mathbf{y}_{w,J}(i)=&\sqrt{P_JL_{aw}g^{(w)}}~\mathbf{x}_J(i)\Bigg[\tilde{h}_{aw,j_1}\sqrt{M_{a,s}}\nonumber\\
&\hspace{2.8cm}+\sum_{\substack{
	j'=1\\
	j'\neq j_1}}^{N_J}\tilde{h}_{aw,j'}\sqrt{m_{a,j'}}\Bigg],
\end{align}
 where $g^{(w)}$ is Willie's beamforming gain, $m_{a,j'}$ is the side lobe gain of the $j'$-th array, and $\tilde{h}_{aw,j'}$ is the fading coefficient from Alice's $j'$-th jamming array to Willie. Assuming that the main lobe gain is much larger that the side lobe gains,
   we can expect the summation term inside the bracket to have negligible contribution compared to the term $\tilde{h}_{aw,j_1}\sqrt{M_{a,s}}$ for \textit{typical} realizations of channel fading coefficients. Note that $N_J$ is relatively small given the wide beamwidths used. Additionally, the fading coefficients $\tilde{h}_{aw,j'}$'s have different phases and hence, the summation term does not blow up with $N_J$. Therefore, we can approximate $\mathbf{y}_{w,J}(i)$ as
 \begin{align}\label{ywj_approx}
 \mathbf{y}_{w,J}(i)\approx\sqrt{P_JL_{aw}M_{a,s}g^{(w)}}~\tilde{h}_{aw,j_1}\mathbf{x}_J(i).
 \end{align}

\textit{2- A single side lobe from jamming arrays to Bob:} Same as in \eqref{ywj}, the received jamming signal by Bob at the $i$-th channel use is expressed as
\begin{align}\label{ybj}
	\mathbf{y}_{b,J}(i)&=\sqrt{P_JL_{ab}g^{(b)}}~\mathbf{x}_J(i)\sum_{j''=1}^{N_J}\tilde{h}_{ab,j''}\sqrt{m_{a,j''}}\nonumber\\
	&\stackrel{(a)}{\approx}\sqrt{P_JL_{ab}m_{a,j_2}g^{(b)}}~\tilde{h}_{ab,j_2}\mathbf{x}_J(i),
\end{align}
where $\tilde{h}_{ab,j''}$ is the fading coefficient from Alice's $j''$-th jamming array to Bob. Moreover, step $(a)$ in \eqref{ybj} is obtained by considering the largest side lobe gain, denoted by $j_2$-th array, as the dominant term of the summation.

Now, given the above two assumptions, it is easy to observe that the performance of the new multi-array system model can be approximated according to our earlier results on the dual-array model. The only difference is that we no longer need to take the average over the gain of the jamming array to Willie to compute $\E[P^*_{e,w}]$ since that gain is deterministically equal to $M_{a,s}$ given the multi-array architecture. All the other derivations remain the same.

\subsection{Future Directions}\label{Sec5C}
{

Given the superiorities of covert communication over the mmWave bands compared to that of RF systems, and that not much work has been done in this area, significant effort is needed to fill the gap on various aspects of covert mmWave communication. In the following, we discuss some possibilities for future research in this direction.
\subsubsection{Uncertainty about Willie's Location}\label{Sec5C1}
In Section \ref{Sec5A}, we explained the importance of obtaining Willie's location information. However, it is desirable to explore how Alice's uncertainty about such information impacts the system performance, e.g., the effective covert rate (see, e.g., \cite{forouzesh2020covert2}). Also, exploring potential approaches that enable obtaining partial information about Willie's location and then characterizing their performance is a viable research direction.

\subsubsection{Precise Performance Characterization of the  Multi-Array System Model}\label{Sec5C2} In Section \ref{Sec5B}, we highlighted how the performance of the multi-array system model proposed in Section \ref{Sec5A} can, \textit{approximately}, be characterized using the analytical results in this paper. One might be able to provide a more rigorous analysis by eliminating the two assumptions made in Section \ref{Sec5B}. To this end, some analytical tools, such as \cite{karagiannidis2001distribution}, to characterize the weighted sum of the involved RVs in \eqref{ywj} and \eqref{ybj} are potentially useful.

\subsubsection{Distribution of the Jamming Signal Power $P_J$}\label{Sec5C3} In this paper, we considered a uniform distribution for $P_J$ in the interval $[0,P_J^{\rm max}]$. Characterizing the performance metrics for the considered covert mmWave communication system model under other statistical distributions for $P_J$ is a straightforward yet important follow-up research that can help determining the optimal/best distribution(s) for the jamming signal power.

{\subsubsection{Covert Communication under Partial Knowledge about $P_J$} In practice, we may be interested in satisfying the covertness requirement for a certain $\epsilon$ and not for all $\epsilon>0$. In that case, one might be able to slightly degrade the covertness requirement with the hope of improving the quality of the Alice-Bob transmission. Characterizing performance-covertness trade-offs and developing efficient schemes to allow Bob achieving some knowledge about the jamming signal (while minimizing Willie's knowledge) are important directions for future research.}

\subsubsection{Covert MmWave Communication under other Potential System Models} In this paper, we have incorporated jamming signals with random realizations per transmission block to enable positive-rate covert mmWave communication in the limit of large blocklengths. One can explore covert mmWave communication under other potential system models, such as uncertainty about the channel gains, noise power, transmission blocks, etc., by utilizing results already established in the literature (see, e.g., \cite{ lee2015achieving,goeckel2016covert,bash2016covert,hu2018covert,wang2019covert,shahzad2017covert,shahzad2018achieving,shahzad2019covert,hu2019covert}).

\subsubsection{Multiple Alice/Bob/Willie} Throughout this paper, we considered the conventional setting of covert communication which consists of a single legitimate transmitter Alice, a single legitimate receiver Bob, and a single warden Willie. Extending the results of the paper to more realistic scenarios, consisting of multiple legitimate transmitters and receivers and multiple wardens (see, e.g., \cite{forouzesh2020covert,chen2021uav,arghavani2021game}), is an important direction for future research. 
}

\section{Numerical Results}\label{Sec6}
In this section, we provide numerical results for various performance metrics delineated in \Tref{col2}, \Tref{thm3}, \Pref{prop4}, and \Tref{Thm5}. The parameters listed in Table \ref{T2} are considered in our numerical analysis unless explicitly mentioned. 
It is assumed that the beamsteering error follows a Gaussian distribution with mean zero and variance $\Delta^2$; hence, $F_{|\mathcal{E}|}(x)={\rm erf}(x/(\Delta\sqrt{2}))$ where ${\rm erf}({\cdot})$ denotes the error function \cite{di2015stochastic}. 
Moreover, the blockage model $P_{\rm LOS}(d_{ij})={\rm e}^{-d_{ij}/200}$ \cite{andrews2017modeling} is used throughout the numerical analysis.

\begin{table}[t]
	\centering
	\caption{Parameters used for the numerical analysis.}
	\label{T2}
	\begin{tabular}{M{1.95in}||M{1.1in}}
		Coefficients & Values\\
		\hline \hline\vspace{0.05cm}
		Link lengths $(d_{aw},d_{ab})$ & $(25,25)$ \si{m}\\\hline\vspace{0.05cm}
		Path loss exponents $(\alpha_{\rm L},\alpha_{\rm N})$ & $(2,4)$\\\hline\vspace{0.1cm}
	Path loss intercepts $(C_{\rm L},C_{\rm N})$ & 	 $(10^{-7},10^{-7})$ \\\hline	\vspace{0.05cm}
Main lobe gains $(M_{a,f},M_{a,s},M_{b})$&$(15,15,15)$ \si{dB}\\\hline\vspace{0.05cm}
Side lobe gains $(m_{a,f},m_{a,s},m_{b})$&$(-5,-5,-5)$ \si{dB}\\\hline\vspace{0.05cm}
Transmit power of Alice's first array, $P_a$ & $20$ \si{dBm}\\\hline\vspace{0.05cm}
Noise power $(\sigma^2_w,\sigma^2_b)$ & $(-74,-74)$ \si{dBm}\\\hline\vspace{0.05cm}
Nakagami fading parameters $(\nu_{\rm L},\nu_{\rm N})$ & $(3,2)$\\\hline\vspace{0.05cm}
Array beamwidths $(\theta_{a,f},\theta_{a,s},\theta_{b})$ & $(30^{\rm o},30^{\rm o},30^{\rm o})$\\\hline\vspace{0.05cm}
Beamsteering error parameter, $\Delta$ & $5^{\rm o}$
	\end{tabular}
\end{table}

\begin{figure}[t]
	\centering
	\includegraphics[width=3.6in]{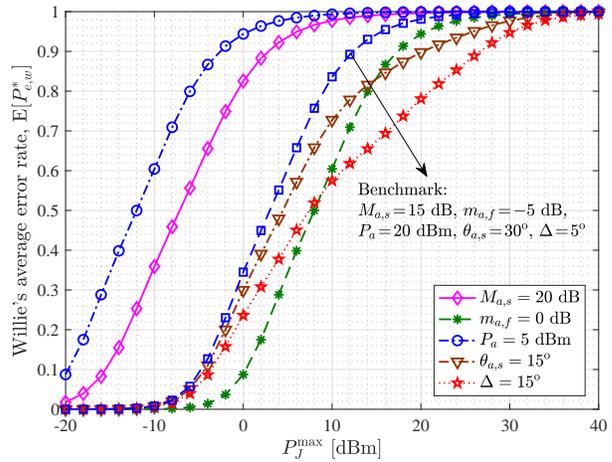}
	\caption{The expected value $\E[P^*_{e,w}]$ of Willie's detection error rate for a benchmark scenario with $M_{a,s}=15$ \si{dB}, $m_{a,f}=-5$ \si{dB}, $P_a=20$ \si{dBm}, $\theta_{a,s}=30^{\rm o}$, and  $\Delta=5^{\rm o}$. The effect of different parameters is explored by considering the values $M_{a,s}=20$ \si{dB}, $m_{a,f}=0$ \si{dB}, $P_a=5$ \si{dBm},  $\theta_{a,s}=15^{\rm o}$, and  $\Delta=15^{\rm o}$  while keeping the rest of the parameters exactly the same as the benchmark scenario.}\label{Fig1}
		\vspace{-0.15in}
\end{figure}
\begin{figure}[t]
	\centering
	\includegraphics[width=3.6in]{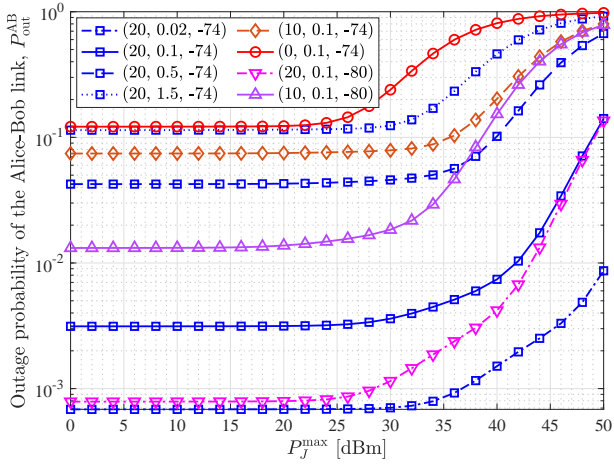}	
	\caption{The outage probability of the Alice-Bob link for  various values of the transmit power $P_a$, threshold rate $R_b$, and noise variance $\sigma^2_b$. The elements of the triples in the legend are $P_a$ in \si{dBm}, $R_b$, and $\sigma^2_b$ in \si{dBm}, respectively.}\label{Fig2}
	\vspace{-0.1in}
\end{figure}
Fig. \ref{Fig1} shows the expected value $\E[P^*_{e,w}]$ of Willie's detection error rate for a benchmark scenario, corresponding to the parameters listed in Table \ref{T2}, as a function of $P_J^{\rm max}$. Moreover, the impact of some relevant parameters, i.e., $M_{a,s}$, $m_{a,f}$, $P_a$, $\theta_{a,s}$, and $\Delta$ is evaluated by changing each of these parameters while keeping the rest of the parameters exactly the same as the benchmark scenario. 
As expected, $\E[P^*_{e,w}]$ monotonically increases with $P_J^{\rm max}$ {since a larger jamming signal will degrade Willie's performance to a greater extent.} Also, reducing $P_a$ degrades Willie's performance since the power level of the desired signal is reduced making it more difficult to be detectable by Willie. Moreover, increasing $M_{a,s}$ deteriorates Willie's performance by exposing his receiver to a more intense jamming signal. 
On the other hand, decreasing $\theta_{a,s}$ or increasing $\Delta$ decrease $\E[P^*_{e,w}]$ since they reduce the probability of Willie receiving the jamming signal with the main lobe of Alice's second array. Finally, increasing $m_{a,f}$ also improves Willie's performance by revealing a higher level of the desired signal, gained by $m_{a,f}$, to Willie.

The outage probability of the Alice-Bob link is illustrated in Fig. \ref{Fig2} for various values of the transmit power $P_a$, threshold rate $R_b$, and noise variance $\sigma^2_b$. The rest of the parameters are the same as those in Table \ref{T2}. As expected, $P_{\rm out}^{\rm AB}$ monotonically increases with $P_J^{\rm max}$. {Moreover, the outage probability increases by increasing the threshold rate $R_b$ since it is harder to guarantee a larger target rate (without outage). Additionally,  the reliability of Alice-to-Bob transmission degrades by increasing the noise variance $\sigma^2_b$ while increasing $P_a$ improves the performance by exposing a higher level of the desired signal to Bob.}

\begin{table}[t]
	\centering
	\caption{Covert rates for $\epsilon=0.05$ and various threshold rates. }
	\label{T3}
	\begin{tabular}{M{0.7cm}||M{0.9cm}M{0.8cm}M{0.8cm}M{0.8cm}M{0.8cm}M{0.8cm}}  
		$R_b$ & $0.1$ & $0.5$& $1$& $2.5$ & $5$ & $10$  \\ \hline\hline \vspace{0.08cm}
		$P_{\rm out}^{*\rm AB}$ & $0.00314$ & $0.04253$ & $0.0935$ & $0.121$ & $0.1308$ & $0.9913$ \\ \hline\vspace{0.08cm}
		$\overline{R}^*_{a,b}$ & $0.0997$ &   $0.4787$ &   $0.9065$  &  $2.1975$  &  $4.3459$ &   $0.0866$
		\\ \hline 
	\end{tabular}
\end{table}
\begin{figure}[t]
	\centering
	\includegraphics[trim=0.8cm 0cm -0.8cm 0cm ,width=3.6in]{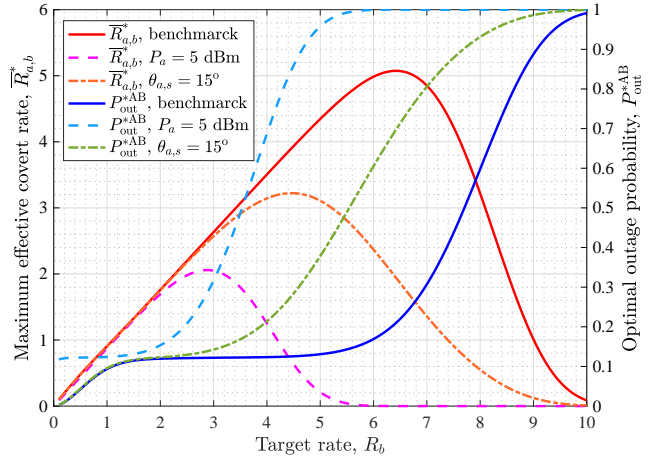}
	\caption{The effective covert rate $\overline{R}^*_{a,b}$ and the corresponding optimal outage probability $P_{\rm out}^{*\rm AB}$ as a function of target rate $R_b$. In addition to the benchmark scenario, two other scenarios of Fig. \ref{Fig1}, namely those obtained by changing $P_a$ from $20$ \si{dBm} to $5$ \si{dBm}, and $\theta_{a,s}$ from $30^{\rm o}$ to $15^{\rm o}$, are also considered.}\label{Fig3}
\end{figure}
Effective covert rates corresponding to the benchmark scenario in Fig. \ref{Fig1} are summarized in Table \ref{T3} for $\epsilon=0.05$ and various threshold rates. To obtain these results, we first numerically solved the equation $\E[P^*_{e,w}]=1-\epsilon$ for $P_J^{\rm max}$, given the parameters corresponding to the benchmark scenario. This resulted in the optimal value of $P_{J,\rm opt}^{\rm max}=15.52$ \si{dBm}. Then we computed the corresponding optimal outage probabilities $P_{\rm out}^{*\rm AB}$, for various target rates, according to \Tref{thm3}. The effective covert rate $\overline{R}^*_{a,b}$ and the corresponding optimal outage probability $P_{\rm out}^{*\rm AB}$ for the considered benchmark scenario is also plotted, as a function of target rate $R_b$, in Fig. \ref{Fig3}. Furthermore, Fig. \ref{Fig3} includes the results of $\overline{R}^*_{a,b}$ and $P_{\rm out}^{*\rm AB}$ for two other scenarios of Fig. \ref{Fig1}, namely those obtained by changing $P_a$ from $20$ \si{dBm} to $5$ \si{dBm}, and $\theta_{a,s}$ from $30^{\rm o}$ to $15^{\rm o}$.
It is observed that, for a given link, the effective covert rate first increases and then decreases by increasing the threshold rate. {This is because, after some point, the outage probability $P_{\rm out}^{*\rm AB}$ quickly transitions from $0$ to $1$.} The maximum effective covert rate that is achievable for the benchmark scenario is $5.0743$ that is obtained for the target rate of $R_b=6.42$ with the corresponding optimal outage probability of $P_{\rm out}^{*\rm AB}=0.2096$. Moreover, maximum effective covert rates of $\overline{R}^*_{a,b}=2.0585$ and $3.2223$ are achievable at the target rates of $R_b=2.88$ and $4.46$ with the corresponding optimal outage probabilities of $P_{\rm out}^{*\rm AB}=0.2853$ and $0.2775$ for the scenarios of $P_a=5$ \si{dBm} and $\theta_{a,s}=15^{\rm o}$, respectively. The optimal values of $P_J^{\rm max}$ for these two scenarios are $P_{J,\rm opt}^{\rm max}=0.52$ \si{dBm} and $25.91$ \si{dBm}, respectively. Although reducing $\theta_{a,s}$ to $15^{\rm o}$ does not directly impact the performance of the Alice-Bob link (e.g., the outage probability or ergodic capacity), it requires much stronger jamming signals with $P_{J,\rm opt}^{\rm max}=25.91$ \si{dBm} to satisfy the covertness requirement which significantly degrades the performance compared to the benchmark scenario. On the other hand, the performance drop-off of the case $P_a=5$ \si{dBm} is a direct consequence of the much lower transmit power used compared to the benchmark scenario though a much weaker jamming signal of $P_{J,\rm opt}^{\rm max}=0.52$ \si{dBm} is enough to satisfy the covertness constraint.

\begin{figure}[t]
	\centering
	\includegraphics[width=3.6in]{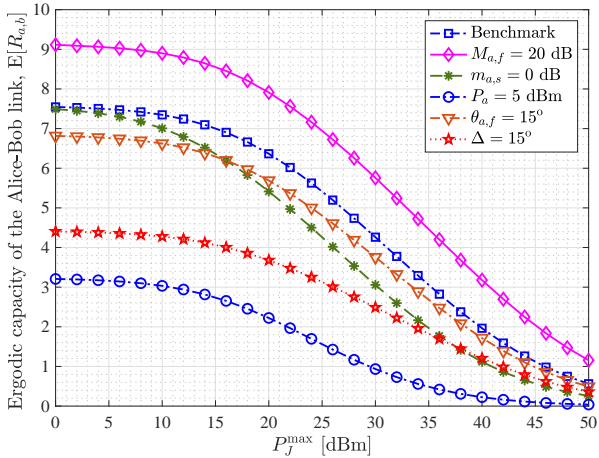}
	\caption{The ergodic capacity $\E[R_{a,b}]$ of the Alice-Bob link for the benchmark scenario, corresponding to the parameters in Table \ref{T2}, and several other setups obtained by changing $M_{a,f}$, $m_{a,s}$, $P_a$,  $\theta_{a,f}$, and $\Delta$.}\label{Fig4}
\end{figure}
The ergodic capacity $\E[R_{a,b}]$ of the Alice-Bob link is shown in Fig. \ref{Fig4} for the benchmark scenario, corresponding to the parameters in Table \ref{T2}. Moreover, the impact of several parameters is examined by changing each one while keeping the reset of the parameters as Table \ref{T2}. As expected, $\E[R_{a,b}]$ monotonically decreases by $P_J^{\rm max}$. { Moreover, enlarging $M_{a,f}$, $P_a$, and $\theta_{a,f}$ positively impacts the ergodic capacity by exposing a higher level of the desired signal to Bob. Additionally, increasing $m_{a,s}$ reduces $\E[R_{a,b}]$ by imposing a stronger jammer on Bob. Finally, increasing $\Delta$ negatively impacts the ergodic capacity by reducing the chance of receiving the desired signal at Bob with a main-lobe gain.}
It is worth mentioning that mmWave links benefit from much larger bandwidths compared to RF links; hence, the results in Figs. \ref{Fig3} and \ref{Fig4} imply much higher data rates, in bits per second, compared to that of RF communication counterparts.

{Finally, it is important to study the performance loss as a result of the existence of the warden Willie. Note that the performance metrics in the absence of Willie can be obtained from the results of the paper by studying the performance of our system model in the limit of $P_J^{\rm max}\to 0$ (i.e., $P_J^{\rm max}\to -\infty$ \si{dBm}). Observe that the outage probability and ergodic capacity results plotted in Figs. \ref{Fig2} and \ref{Fig4}, respectively, as a function of $P_J^{\rm max}$, start at some constant values and then change once the value of $P_J^{\rm max}$ is large enough to impact the performance of the Alice-Bob link. Therefore, the constant values of these plots at very small values of $P_J^{\rm max}$ correspond to the outage probability and ergodic capacity of the system in the absence of Willie. Hence, Figs. \ref{Fig2} and \ref{Fig4} clearly illustrate the performance loss due to the existence of Willie as a function of $P_J^{\rm max}$. In fact, the outage probability increases and the ergodic capacity decreases as we increase $P_J^{\rm max}$ to support a stronger level of covertness due to the existence of the warden.}

\section{Conclusions}\label{Sec7}
In this paper, we investigated covert communication over mmWave links. We employed a dual-beam transmitter to simultaneously transmit the desired signal to the destination and propagate a jamming signal to degrade the warden's  performance. We characterized Willie's detection error rate and the closed-form of its expected value from Alice's perspective. We then derived the closed-form expression for the outage probability of the Alice-Bob link which enabled us to formulate the optimal achievable covert rates. We further obtained tractable forms for the ergodic capacity of the Alice-Bob link involving only one-dimensional integrals that can be computed in closed forms for most ranges of the channel parameters.
Moreover, we elucidated how the results can be extended to more practical scenarios, taking into account the uncertainty about Willie's location. We also highlighted several interesting directions for future research on covert mmWave communication. Through comprehensive numerical studies, we analyzed the behavior of the derived performance metrics with respect to variety of channel and system parameters. Our results demonstrated the advantages of covert mmWave communication compared to the RF counterpart, calling for further research on this novel area.

\appendices
\section{A Useful Lemma for the Integration over ${\rm Ei}(\cdot)$}\label{App_A}
In \cite[Lemma 1]{jamali2019uplink}, a useful lemma is proved for the integral of $\int_{c_1}^{c_2}{\rm e}^{bx}{\rm Ei}(ax)dx$ with $c_1,c_2>0$, $a<0$, and $b\in\mathbb{R}$ such that (s.t.) $a+b<0$. In this appendix, we prove that the same result, with a slight change, can be applied to the case of $b=-a$, i.e., $a+b=0$ (see, e.g., \cite[Eq. (06.35.21.0014.01)]{wolfram}).
\begin{lemma}\label{lemma_app1}
For any $c_1,c_2>0$ and $a<0$, we have
\begin{align}\label{app1_1}
\int_{c_1}^{c_2}{\rm e}^{-ax}{\rm Ei}(ax)dx=&\frac{1}{-a}\Big[{\rm e}^{-ac_2}{\rm Ei}(ac_2)\nonumber\\
&-{\rm e}^{-ac_1}{\rm Ei}(ac_1)-\ln\left({c_2}/{c_1}\right)\Big].
\end{align}
\end{lemma}
\begin{proof}
Note based on \cite[Lemma 1]{jamali2019uplink} that for $c_1,c_2>0$, $a<0$, and $b\in\mathbb{R}$ s.t. $a+b<0$, we have
\begin{align}\label{app1_2}
\int_{c_1}^{c_2}{\rm e}^{bx}{\rm Ei}(ax)dx=\frac{1}{b}\left[{\rm e}^{bt}{\rm Ei}(at)-{\rm Ei}([a+b]t)\right]\!{\Big |}_{c_1}^{c_2},
\end{align}
where $f(t)|_{c_1}^{c_2}\triangleq f(c_2)-f(c_1)$ for the function $f(t)$. In the case of $b=-a$ per \Lref{lemma_app1}, the argument of the second exponential integral function ${\rm Ei}([a+b]t)$ in \eqref{app1_2} is zero. Based on \cite[Eq. (1)]{harris1957tables}, $\lim_{x\to 0}{\rm Ei}(x)=\gamma+\ln|x|$, where $\gamma=0.57721$ is the Euler's constant. Therefore, we can write ${\rm Ei}([a+b]t){\big |}_{c_1}^{c_2}$ for the case of $a=-b$ as
\begin{align}\label{app1_3}
\!\!{\rm Ei}([a+b]t){\big |}_{c_1}^{c_2}\!=\!\lim_{x\to 0} {\rm Ei}(xt){\big |}_{c_1}^{c_2}\!=\!\lim_{x\to 0} \ln\!\left(\!\frac{|xc_2|}{|xc_1|}\!\right)\!\!=\!\ln\!\left(\!\frac{c_2}{c_1}\!\right)\!.\!
\end{align}
This together with some similar arguments as the proof of \cite[Lemma 1]{jamali2019uplink} completes the proof of \Lref{lemma_app1}.
\end{proof}


\end{document}